\documentclass[10pt]{amsart}
\usepackage{amssymb,amsmath,mathtools}
\usepackage{enumerate}
\usepackage{color}
\usepackage{caption}

\newtheorem{theorem}{Theorem}[section]
\newtheorem{proposition}[theorem]{Proposition}
\newtheorem{lemma}[theorem]{Lemma}
\newtheorem{corollary}[theorem]{Corollary}
\theoremstyle{definition}
\newtheorem{definition}[theorem]{Definition}
\newtheorem{assumption}[theorem]{Assumption}
\newtheorem{example}[theorem]{Example}
\newtheorem{remark}[theorem]{Remark}

\numberwithin{equation}{section}

\begin{document}

\title[Convergence rates for the estimation of risk measures]{Non-asymptotic convergence rates for the plug-in estimation of risk measures}

\author{Daniel Bartl \and Ludovic Tangpi}
\address{Vienna university, department of mathematics}
\email{daniel.bartl@univie.ac.at}
\address{Princeton University, ORFE}
\email{ludovic.tangpi@princeton.edu}
\date{\today}

\begin{abstract}
	Let $\rho$ be a general law--invariant convex risk measure, for instance the average value at risk, and let $X$ be a financial loss, that is, a real random variable.
	In practice, either the true distribution $\mu$ of $X$ is unknown, or the numerical computation of $\rho(\mu)$ is not possible. 
	In both cases, either  relying on historical data or using a Monte-Carlo approach, one can resort to an i.i.d.\ sample of $\mu$ to approximate $\rho(\mu)$ by the finite sample estimator $\rho(\mu_N)$ (where $\mu_N$  denotes the empirical measure of $\mu$).
	In this article we investigate convergence rates of  $\rho(\mu_N)$ to $\rho(\mu)$.
	We provide non-asymptotic convergence rates for both the deviation probability and the expectation of the estimation error.
	The sharpness of these convergence rates is analyzed.
	Our framework further allows for hedging, and the convergence rates we obtain depend neither on the dimension of the underlying assets, nor on the number of options available for trading.
\end{abstract}

\maketitle
\setcounter{equation}{0}


\section{Introduction}
\label{sec:intro}

Risk is a pervasive aspect of the financial industry as every single financial decision carries a certain amount of risk. 
Correctly quantifying riskiness is therefore of central importance for financial institutions.
The idea is often to consider the profit and loss $F(S)$ resulting from an investment in assets $S$.
A fundamental innovation (that can be traced back to the work of H.~Markovitz \cite{markowitz1952} in the 1950s) allowing to quantify the risk of $F(S)$ was the introduction of the concept of risk measures, which allows to assign  a \emph{numerical value} $\rho(F(S))$ to the profit and loss $F(S)$ depending on the agent's risk aversion.
In other words, one can focus on a single number to make decisions rather than on the whole distribution of the loss.
As a consequence, computing $\rho(F(S))$ becomes an essential task for the risk manager.

For a long time, the value at risk (VaR) has been the industry standard for risk management.
As a result, the numerical simulation of VaR (i.e.\ of quantiles) is well-understood, and various methods can be found in \cite{Glass-Hei-Shah,Jorion,Hong09,Jin-Fu-Xiong2003} and in their references.
However, there are many criticisms for the VaR\footnote{See for instance McNeil et.~al.\ \cite[Section 2.2]{McNeil15} for ample discussion.}, so much so that the Basel Committee on Banking Supervision which oversees risk management for financial institutions has recommended since 2013 to use expected shortfall (also known as average value at risk (AVaR) or conditional value at risk) as the benchmark risk measure \cite{Basel}.

Intuitively, the AVaR at level $u\in (0,1)$ can be understood as the average of all $\mathrm{VaR}_v$ over $v \in (u,1)$.
Thus, it does not only take into account occurrence of large losses, but also their size.
The estimation of AVaR (in the context of portfolio optimization) was for instance considered by Rockafellar and Uryasev \cite{Rock-Ury00} who fundamentally used the fact that the AVaR at level $u$ of the loss $F(S)$ can be written as
\begin{equation}
\label{eq:avar.oce.def}
	\mathrm{AVaR}_u(F(S)) = \inf_{m\in \mathbb{R}} \Big( \frac{1}{1-u}E[(F(S)-m)^+] + m\Big).
\end{equation}
This representation shows, in particular, that AVaR is almost \emph{risk neutral} for very large losses as it is linear in the tails.
This prompted the generalization to convex risk measures that behave nonlinearly both in the tail and the center of the distribution, including for instance the optimized certainty equivalent (OCE) obtained by replacing the function $x\mapsto \frac{1}{1-u}x^+$ in \eqref{eq:avar.oce.def} by a convex loss function $l\colon\mathbb{R}\to\mathbb{R}$, see Ben-Tal and Teboulle \cite{Ben-Teb,ben-tal02}, or the shortfall risk measure (SF) defined in a similar spirit, see F\"ollmer and Schied \cite{foellmer02} (and Section \ref{sec:main.avar.etc} for details).

More generally, a rigorous unifying approach to risk management was initiated by Artzner et.~al.\ \cite{Artzner1999} and matured into an impressive theory of risk measures.
We refer for instance to the monographs \cite{McNeil15,foellmer01} for excellent expositions.
A general convex risk measure is defined as follows:
\begin{definition}[Convex risk measure]
\label{def:risk.measure.l.infty}
	A functional $\rho\colon L^\infty\to\mathbb{R}$ over a standard probability space is a convex risk measure\footnote{
	Observe that in contrast to the original definition  \cite{Artzner1999,foellmer01}, we take risk measures to be increasing.
	This means that $X$ models the (discounted) loss and $\rho(X)$ is the capital to be added to a position with loss $X$ to make it acceptable, see e.g.~\cite[Chapter 6]{McNeil15} for a similar framework and more details.
	This is done for notational convenience and does not affect generality. In fact, putting $\tilde\rho(X) := \rho(-X)$, the functional $\tilde\rho$ is a risk measure in the sense of \cite{Artzner1999,foellmer01}.
	} if
	\begin{enumerate}[(a)]
	\item $\rho(X+m)=\rho(X)+m$ for all $X$ and $m\in\mathbb{R}$ and $\rho(0)=0$,
	\item $\rho(X)\leq \rho(Y)$ if $X\leq Y$ almost surely,
	\item $\rho(\lambda X + (1-\lambda) Y) \leq \lambda \rho(X)+ (1-\lambda)\rho(Y)$ for $\lambda\in[0,1]$.
	\end{enumerate}
\end{definition}

The monotonicity condition (b) is natural, and models preference for more profits.
The condition (a) known as cash-invariance or translation invariance stems from the desire to interpret $\rho(X)$ as a capital requirement.
That is, the minimal cash value which, if added to the position $X$ would make it acceptable for regulators.
The convexity property (c) means that more diversified positions should be less risky.
In practice, for numerical simulation, it is more convenient to work with the distribution of the loss rather than its observed realization.
Therefore, it is often assumed that risk measures are \emph{law--invariant} (or \emph{law--determined}\footnote{Actually ``law--determined'' might be a more sensible term to describe this property, but in accordance the literature we will use the term ``law--invariant'' since it is predominantly used.}), meaning that
\begin{enumerate}[(d)]
	\item $\rho(X)=\rho(Y)$ if $X\sim Y$, that is, if $X$ and $Y$ have the same distribution.
\end{enumerate}
Observe that most examples of risk measures fulfill this condition.
We make the convention that, throughout this paper, the term `risk measure' always refers to a convex law--invariant risk measure.
Consequently, we will use the shorthand notation 
\[\rho^\mu(F):= \rho(F(S))\quad \text{where}\quad S\sim \mu,\]
that is, $\rho^\mu(F)$ is the risk of $F(S)$ computed according to the risk measure $\rho$ when $S$ has the distribution $\mu$.

\vspace{0.5em}

The numerical computation of a law--invariant risk measure depends on the probability distribution $\mu$ of $S$.
For AVaR for instance, one issue is to efficiently approximate the integrals $\int_{\mathbb{R}} (F(x) - m)^+\,\mu(dx)$ (assuming $S$ is real-valued).
In some cases, this integral operation is computationally costly.
Moreover, in many practical applications the distribution $\mu$ is not precisely known.
A natural idea is to approximate the integral by the sample average $\frac1N\sum_{n=1}^N(F(S_n)-m)^+$ where $S_1,\dots,S_N$ are independent random variables with distribution $\mu$, and the minimization over $m$ (see \eqref{eq:avar.oce.def}) can be reduced to linear programming, see \cite[Section 3]{Rock-Ury00}.
When $\mu$ is not known, this Monte-Carlo simulation can be carried-out on historical data.
\begin{example}
For instance in the context of portfolio optimization, $S$ is a vector of $d$ stocks returns $S:=(S_1^i - S^i_0)_{i=1,\dots,d}$ and $F(S)$ takes the form $F(S):=\sum_{i=1}^d g_i(S_1^i-S_0^i)$ 
where $(g_1,\dots,g_d)$ are portfolio weights.
Strictly speaking, in practice the time series formed by historical returns is of course not i.i.d. 
It shows patterns of changing volatility.
Some workarounds in the literature include working with longer interval returns series, see \cite[Section 4.1]{McNeil15} or using a semi--parametric approach to estimate returns as $\Delta S_i := X_i Z_i$ where $X_i$ is a diagonal volatility matrix modeled by a GARCH model and where the innovation processes $Z_i$ are i.i.d., see \cite[Chapter 2]{McNeil15}.
Thus, we can make the simplifying assumption that we are working with i.i.d.\ observations.
In both cases, denoting by 
\[\mu_N:=\frac1N\sum_{n\leq N}\delta_{S_n} \quad\text{where }  S_1,\dots,S_N\sim S \text{ i.i.d.}\]
the empirical distribution of the $N$ observations $S_1,\dots, S_N$, $\text{AVaR}^{\mu_N}(F)$ is a \emph{non-parametric} finite sample estimator of $\text{AVaR}^{\mu}(F)$.
\end{example}
This idea extends to general risk measures and arbitrary functions $F$:
\begin{definition}[Plug-in estimator]
	For every $N\geq 1$, denote by
	\begin{equation*}
	\rho^{\mu_N}(F) = \rho(F(\widehat S)) \quad\text{where } \widehat S\sim \mu_N
	\end{equation*}
	the \emph{plug-in estimator} of $\rho^\mu(F)$.
\end{definition}

As we will soon observe, this estimator\footnote{The issue of numerical simulation of the estimator $\rho^{\mu_N}(F)$ is considered in \cite{Chu-Tangpi21}.} is \emph{consistent} (see Corollary \ref{cor:oce.consistency}) but typically \emph{underestimates} the true risk $\rho^\mu(F)$ (see Remark \ref{rem:bias}).
The latter observation is consequential from the practical standpoint.
In fact, the idea of risk measures is precisely to protect oneself from risky investments, thus underestimating the risk of an asset is precisely what risk managers want to avoid.
Since the finite sample estimator $\rho^{\mu_N}(F)$ is the most natural in non-parametric estimations and is widely used in practice, it is therefore essential to understand just how much it underestimates the true risk.

Thus the question for the risk manager is:
\begin{align*}
&\text{How \emph{far} is } \rho^{\mu_N}(F) \text{ from }  \rho^\mu(F) \text{ for a fixed sample size  } N?
\end{align*}
This is an essential question because its answer will give theoretical insights allowing risk managers to parsimoniously use data.
To make the question rigorous, one of course needs to give a meaning to `\emph{far}', as the estimation error $|\rho^{\mu_N}(F)-\rho^\mu(F)|$ is random (it depends on the observations from $S$).

The goal of this article is to answer the above question by providing \emph{non-asymptotic} convergence rates on the expected estimation error and on the probability that the estimation  error exceeds some prescribed threshold.
The difference between \emph{asymptotic} rates and \emph{non-asymptotic} rates should be underscored: while there are instances where the asymptotic rates suggest a much faster convergence, this is only true within the asymptotic regime.
In particular, relying on asymptotic rates e.g.\ for the computation of the needed sample size $N$ to guarantee that the estimator preforms well with a certain confidence might give a far too optimistic number.
Non-asymptotic rates however hold for every $N$, and give an order of magnitude of the sample size $N$ needed to achieve a desired estimation accuracy.

\vspace{0.5em}

\emph{Our main results} show the following:
for a general risk measure, the usual $1/\sqrt{N}$ convergence rate dictated by the central limit theorem needs not to hold true.
We introduce a simple and tractable notion of regularity for risk measures (quantified by a parameter $q\in(1,\infty)$) and show that this notion of regularity governs the convergence rates.
In more precision, if a risk measure is regular with parameter $q$, then 
	\begin{align}
	\label{eq:rates.intro}
	\begin{split}
	E^\ast\Big[\big|\rho^{\mu_N}(F)-\rho^{\mu}(F)\big|\Big]
	&\leq \frac{C}{\sqrt{N}^q} \\
	P^\ast\Big[\big|\rho^{\mu_N}(F) - \rho^{\mu}(F)\big|\geq \varepsilon \Big]
	&\leq C\exp\Big(-cN  \varepsilon^{2q} \Big)
	\end{split}
	\end{align}
for all $N\geq 1$ and $\varepsilon>0$, where $c,C$ are two positive constants depending on $\rho$ and (the $L^\infty$-norm of) $F$; see Theorem \ref{thm:main.hedging.on.Linfty.intro}.
Here $E^\ast$ and $P^\ast$ denote the outer expectation and outer probability, respectively, see, e.g., \cite{van1996weak}.
They are used since $\rho^{\mu_N}$ does not necessarily need to be measurable.
Notably, we show that these rates are sharp (at least up to a factor of $2$); see Proposition \ref{prop:non.asymptotic.no.rates.into}.
While the above represents the major part of this work, we also consider three explicit risk measures separately (the average value at risk, the optimized certainty equivalent, and the shortfall risk) and show that then the usual $1/\sqrt{N}$  convergence rate dictated by the central limit theorem can be recovered in a non-asymptotic fashion, namely \eqref{eq:rates.intro} holds true for all $N\geq 1$ and $\varepsilon>0$ with $q=1$.

In practice, there is much more to risk management than computing the numerical value $\rho(F(S))$ for a given loss $F(S)$: 
in many situations risk managers additionally hedge their exposure to $F(S)$ by investing in the stock market, resulting in a risk based (super-)hedging problem or a utility maximization problem.
Another important part of our work focuses on the investigation of these type of problems and we show that the same rates of convergence as above remain valid in this more general situation.

\vspace{0.5em}
The reader familiar with the theory of sample average approximation (SAA) will likely recognize the average value at risk \eqref{eq:avar.oce.def} as a stochastic programming problem.
Finite sample approximation of such problems have been extensively studied, at least as far as \emph{asymptotic convergence rates} are concerned. 
We give a few references below. 
Note however that, in the generality of Definition \ref{def:risk.measure.l.infty}, risk measures cannot typically be written as the value of a stochastic programming problem (using Kusuoka's theorem we can only write them as a maximum of more and more irregular  stochastic programs).
This substantially complicates the analysis.
The main contribution of the paper is to consider this general situation by balancing the irregularity of the stochastic programs by means of the regularity of the risk measure.
To obtain finite sample rates of convergence, we employ techniques from empirical process theory, specifically Dudley's theorem.
This allow us to consider the above mentioned case of additional hedging and yields rates that do not depend on the dimension / number of hedging instruments.
Finally, let us stress that $F$ (as well as the hedging instruments) are not subject to any continuity condition.

\subsection{Related literature}
The estimation of risk measures is an essential question in quantitative finance, and as such has received a lot of attention, we refer for instance to the monograph of McNeil, Frey and Embrechts \cite{McNeil15} for an in-depth treatment.
See also the book of Glasserman \cite[Chapter 9]{Glass04} for the case of (average) value-at-risk.
In mathematical finance, there is a growing interest on statistical aspects of quantitative risk management.
See Embrechts and Hofert \cite{Emb-Hof14} for an excellent review of the main lines of research in this direction.
Concerning statistical estimation of risk measures, one of the earliest work is that of Weber \cite{Weber07} who considered the problem of estimating $\rho^\mu(F)$ in an asymptotic fashion as $N\to\infty$.
By means of the theory of large deviations, he showed that if $\rho$ is sufficiently regular, then $\rho^{\mu_N}(F)$ satisfies a large deviation principle.
Along the same lines, \cite{Beu-Zaeh10,belomestny2012central,Chen08} obtained central limit theorems for $\rho^{\mu_N}(F)$; see also \cite[Chapter 6]{ADR-2014}.

Aside from large deviation and central limit theorem results, several authors have investigated estimation of specific risk measures and (super)hedging functionals.
These include Pal \cite{pal2006capital,pal2007computing}  who analyzes hedging under risk measures which can be written as the finite maxima of expectations.
Let us further refer to \cite{Pitera-Sch,Krae-Schield-Zaehle15,Guigues2018,Holzmann2020} for other (asymptotic) estimation results, mostly for the average value at risk and expectiles, and under some assumptions on the distribution $\mu$; see e.g.\  Hong, Hu and Liu  \cite{hong2014monte} for a review.
A deviation-type inequality for the value at risk is proposed by Jin, Fu and Xiong \cite{Jin-Fu-Xiong2003}.
The problem of strict superhedging was recently considered by Ob{\l}{\'o}j and Wiesel \cite{obloj2018statistical}, and the problem of portfolio optimization under heavy tails by Bartl and Mendelson \cite{bartl2022monte}.

When the estimation of $\rho^\mu(F)$ is performed repeatedly or periodically, it is  important that the estimator $\rho^{\mu_N}(F)$ be stable, i.e.\ insensitive to small changes of $\mu_N$.
Such insensitivity is often referred to as \emph{robustness} of the risk measure and was first analyzed by Cont, Deguest and Scandolo \cite{cont2010robustness} who investigated a concept of robustness essentially equivalent to continuity of $\rho$ w.r.t.\ weak convergence of measures.
Alternative approaches to robustness were later proposed and analyzed by  Claus, Kr{\"a}tschmer, Schied, Schultz and Z\"ahle \cite{Krae-Sch-Zaeh14,Kraet-Zaehhle17,Krae-Sch-Zaeh2012,Claus2017,Krae-Schie-Zaehle17}.
Along the same lines some authors have investigated risk measures (and other stochastic maximization problems) under model uncertainty to account for the effect of possible misspecification of the estimated model, see e.g.\ \cite{Roboptim,Esfahani,Eck-Kup-Pohl}
 where it is often assumed that the true model belongs to a Wasserstein ball.

Beyond estimation of risk measures, a rich literature in operations research is devoted to the estimation of the value of stochastic optimization problems similar to OCE through the empirical distribution of the underlying probability measure.
This technique goes under the name sample average approximation as mentioned above.
The bulk of the literature in this direction is concerned with convergence issues and questions related to computational complexity of the estimators, see e.g.\ \cite{Kleywegt2001,Bertsimas2017} and the book chapter \cite{Kim-Pas-Hen} for a recent overview.
We also refer to the recent preprint \cite{Kraetschmer21} for asymptotic estimation results as well as error bound estimations using empirical process theory.

\subsection{Organization of the rest of the paper}

We start by presenting the main results of this article in the next section.
The proofs of the moment bounds will be given in Section \ref{sec:oce} for special cases of risk measures and in Section \ref{sec:general.rm}, the main part of the paper where convergence rates for general risk measures are proved.
In these sections, we will also state generalizations of our results to unbounded cases.
The deviation inequalities will be proved in 
Section \ref{sec:deviation}.
Sharpness of the rates for general risk measures is discussed in Section \ref{sec:sharpness} and all remaining proofs are presented in Section \ref{sec:aux}.
The paper ends with an appendix on the theory of empirical processes. 

\section{Main results}
\label{sec:main.results.}

Before presenting our main results, let us generalize the setting of the introduction to the more practically relevant situation where the risk manager can offset the risk from $F(S)$ by trading.
Henceforth, $\mu$ denotes the distribution of $S$ which is a probability measure on a Polish space $\mathcal{X}$ and $\mu_N$ denotes the empirical measure of $\mu$ build from an i.i.d.\ sample $(S_n)_{n\leq N}$ defined on some abstract probability space $(\Omega,\mathcal{F},P)$.
Moreover, $F\colon\mathcal{X}\to\mathbb{R}$ is a measurable function.
In fact, we can additionally consider (measurable) options $G_1,\dots,G_e\colon\mathcal{X}\to\mathbb{R}$ available for \emph{trading} without loss of generality at price zero (where $e\in\mathbb{N}$).
Trading according to a strategy $g\in\mathbb{R}^e$ then yields the outcome $F+\sum_{i=1}^e g_i G_i$.
Thus, assuming the interest rate to be zero throughout, the risk manager's task is to estimate the minimal risk incurred when trading in the option market, that is, to compute
\[ \pi^\mu(F):=\inf_{g\in\mathcal{G}} \rho^\mu\Big(F + \sum_{i=1}^e g_iG_i\Big),\]
where $\mathcal{G}\subset\mathbb{R}^e$ is the set of all admissible trading strategies.
Loosely speaking, the goal here is to `absorb' extreme outcomes of $F$ by trading.
For instance, $\mathcal{G}=\{ g \in[0,1]^e: g_1+\cdots +g_e=1\}$ corresponds to \emph{portfolio optimization}; see \cite{ADR-2014} for some background.
Notice that if $0$ is the only admissible trading strategy, i.e.\ $\mathcal{G} = \{0\}$, then we have $\pi^\mu=\rho^\mu$ and hence all results obtained for $\pi$ translate to $\rho$ as introduced in the previous section.

In an effort to simplify the presentation in this section, we state the results for risk measures defined on bounded random variables and assume throughout this section that $F$ and $G$ are all bounded functions.
In the later sections we partially replace boundedness by integrability assumptions, at the cost of more involved notation.
Moreover, $\mathcal{G}\subset\mathbb{R}^e$ is assumed to be a bounded\footnote{That is, there is a constant $C_\mathcal{G}>0$ such that $\sup_{x\in \mathcal{G}}|x|\le C_\mathcal{G}$, where $|\cdot|$ denotes the Euclidean norm on $\mathbb{R}^e$.} set throughout this article, and this assumption can quickly be checked to be necessary (see Proposition \ref{prop:G.needs.to.be.bounded}).
In order to avoid  discussions regarding measurability issues, we assume throughout this article that $\mathcal{G}$ is a countable set.
As explained in Remark \ref{rem:G.countable}, this assumption can actually be made without loss of generality.

\subsection{Results for general convex risk measures}
Let us now present our main results pertaining to the estimation of law--invariant risk measures in the generality of Definition \ref{def:risk.measure.l.infty}.
As it will become more and more apparent throughout this article, it is necessary to impose some form of continuity assumption on the risk measure in order to derive non-asymptotic convergence rates, see Proposition \ref{prop:non.asymptotic.no.rates.into} and Remark \ref{rem:fatou.lebesque} below.

In order to start with a positive result, we define the notion of continuity that we require right away and discuss its rationale and necessity afterwards.

\begin{definition}[$q$-Regularity]
\label{def:pareto}
	For $q\in(1,\infty)$ a convex risk measure is said to be $q$-\emph{regular} if it satisfies
	\[\sup_{n\in\mathbb{N}} \rho(X\wedge n)<\infty \]
	for all random variables $X$ following a Pareto distribution with shape parameter $q$.
\end{definition}

Recall here that a random variable $X$ has Pareto distribution with scale parameter $x>0$ and shape parameter $q>0$ if 
\[ P[X\ge t] = \begin{cases}
(x/t)^q &\text{if } t\ge x,\\
1 &\text{if } t< x.\end{cases}\]

As the name suggests, $q$-regularity reflects a certain notion of continuity and the familiar reader may recognize it to be stronger than the two classical notions of regularity for risk measures, namely the Fatou property and the Lebesgue property.
The latter two are, however, not enough to guarantee \emph{any} convergence rates, see Section \ref{sec:sharpness}.

The following is the main result of this article.

\begin{theorem}[Rates for general risk measures]
\label{thm:main.hedging.on.Linfty.intro}
	Let $q\in(1,\infty)$ and let $\rho\colon L^\infty\to\mathbb{R}$ be a $q$-regular risk measure.
	Then there are constants $c,C>0$ such that the following hold.
	\begin{enumerate}[(i)]
	\item
	We have the moment bound
	\[ E^\ast\Big[\big|\pi^{\mu}(F)-\pi^{\mu_N}(F)\big|\Big]
	\leq \frac{C}{N^{1/2q}} \]
	for all $N\geq 1$.
	\item
	We have the matching deviation inequality
	\[ P^\ast\Big[\big|\pi^{\mu}(F)-\pi^{\mu_N}(F)\big|\geq \varepsilon \Big]
	\leq C\exp\Big(-cN  \varepsilon^{2q} \Big)\]
	for all $N\geq 1$ and all $\varepsilon>0$.
	\end{enumerate}
\end{theorem}

An important observation is that throughout this paper the rates will never depend on the number $e$ of options, nor on the `dimension' of the underlying space $\mathcal{X}$.
The constants $c$ and $C$ depend on $\rho$, the maximal range of $F$, $G$, the number of options $e$ and the  
 maximal Euclidean norm in $\mathcal{G}$.

\vspace{0.5em}
As an immediate consequence of Theorem \ref{thm:main.hedging.on.Linfty.intro} part \eqref{item:oce.intro.deviation} and the Borel--Cantelli lemma, we obtain that $\pi^{\mu_N}(F)$ is a strongly consistent estimator of $\pi^\mu(F)$.
Note that, under the assumption that $F$ and $G$ are additionally continuous, this is a trivial consequence of weak continuity of $\nu\mapsto \pi^\nu(F)$ and weak  convergence of the empirical measure to the true one.
However, this reasoning does not apply in the present setting as $F$ and $G$ are merely measurable and thus $\nu\mapsto \pi^\nu(F)$ can be discontinuous.

\begin{corollary}[Consistency]
\label{cor:oce.consistency}
	In the setting of Theorem \ref{thm:main.hedging.on.Linfty.intro} we have that
	\[ \lim_{N\to\infty} \pi^{\mu_N}(F) = \pi^{\mu}(F) \]
	$P^\ast$--almost surely.
\end{corollary}

\begin{remark}
\label{rem:tail.integral}
	An interesting by--product of the deviation inequality given in Theorem \ref{thm:main.hedging.on.Linfty.intro} is that it allows to give a non--asymptotic estimation of the error in the $L^p$ norm.
	Indeed, using the tail--integration $E^\ast[|X|^p]=p\int_0^\infty x^{p-1}  P^\ast[|X|\geq x] \,dx$, it follows from part (ii) in Theorem \ref{thm:main.hedging.on.Linfty.intro} that
	\begin{align*} 
		E^\ast\Big[\big|\pi^{\mu}(F)-\pi^{\mu_N}(F)\big|^p\Big]^{1/p} \leq C \frac{ \sqrt{p} }{N^{1/2q}}.
	\end{align*}
	for every $p\geq 1$.
\end{remark}

Let us now come back to the notion of regularity in Definition \ref{def:pareto} and explain both its rationale and necessity. 
This is easiest done with the example of following two risk measures: $\rho_{\max}(X):=\mathop{\mathrm{ess.sup}} X$ and $\rho_{\mathrm{mean}}(X):= E[X]$. 
Then $\rho_{\mathrm{mean}}^{\mu_N}(F)$ is just the empirical mean of $F(S)$ and thus convergence happens at the usual rate $1/\sqrt{N}$.
On the other extreme of the spectrum, $\rho_{\max}^{\mu_N}(F)$ equals the empirical $1$-quantile and it is well known that without very specific assumptions, convergence may happen at arbitrarily slow speed; the unfamiliar reader may skip to Remark \ref{rem:rates.without.lebesque}.
A simple observation pertaining to the source of this different behavior is that small changes of $X$ will result in small changes of $\rho_{\mathrm{mean}}$, while this is not the case for $\rho_{\max}$.
Indeed, changes of $X$ on almost negligible sets can result in significant changes of $\rho_{\max}(X)$ and (unfortunately) a random sample cannot properly exhibit almost negligible events.
From this perspective, $\rho_{\mathrm{mean}}$ is very regular (and indeed, it is $q$-regular for every $q\in(1,\infty)$) while $\rho_{\max}$ is not regular at all (and indeed, it lacks $q$-regularity for any $q\in(1,\infty)$).

While the above discussion focused only on two very extreme risk measures, it happens that Definition \ref{def:pareto} actually interpolates between these two examples.
Indeed, the proposition below shows that the rates obtained in Theorem \ref{thm:main.hedging.on.Linfty.intro} are optimal, at least up to a factor of $2$.

\begin{proposition}[Sharpness of rates]
\label{prop:non.asymptotic.no.rates.into}
	Let $q\in(1,\infty)$ and assume that $F$ takes (at least) two distinct values.
	Then there is a coherent law--invariant risk measure $\rho\colon L^\infty\to\mathbb{R}$ which is $(q+\varepsilon)$-regular for every $\varepsilon>0$ and a constant $c>0$ such that
	\[ \sup_{\mu } E\Big[\big|\rho^\mu(F)-\rho^{\mu_N}(F)\big|\Big] \geq \frac{c}{N^{1/q}} \]
	for all 
	$N\geq 1$.
\end{proposition}

Currently, the authors do not know whether Proposition \ref{prop:non.asymptotic.no.rates.into} can be improved to show that the rates obtained in Theorem \ref{thm:main.hedging.on.Linfty.intro} are actually sharp (i.e.\ whether Proposition \ref{prop:non.asymptotic.no.rates.into} holds with $N^{-1/2q}$ instead of $N^{-1/q}$).
One indication that this might be true is the following:
for $q\approx 1$ the lower bound of Proposition \ref{prop:non.asymptotic.no.rates.into} is approximately $1/N$ but we already know that the actual best possible rate is $1/\sqrt{N}$, as is dictated by the central limit theorem; see Section \ref{sec:sharpness} for a short discussion.
That is, for $q\approx 1$ the lower bound is off exactly by the factor of two.

\vspace{0.5em}
Let us conclude this subsection with a comment regarding the proof of Theorem \ref{thm:main.hedging.on.Linfty.intro}.
As already mentioned, it builds upon empirical processes theory; specifically Dudley's entropy integral theorem (see Appendix \ref{sec:empirical.processes}).
One could, however, wonder whether, the statements of Theorem \ref{thm:main.hedging.on.Linfty.intro} (and Theorem \ref{thm:main.oce.hedging.intro}) would follow from some rather simple to obtain continuity in Wasserstein distance of $\mu\mapsto \rho^\mu(F)$ in combination with convergence rates of empirical measure in Wasserstein distance -- at least if $\mathcal{X}=\mathbb{R}^d$ and $F, G$ are Lipschitz continuous.
While this technique certainly works for dimension $d=1$, in the present general, multidimensional setting this approach would force the convergence rates to be significantly worse:
In dimension $d\geq 3$, the Wasserstein distance converges with rate $N^{-1/d}$  see \cite{Fou-Gui15}.
Thus, even for $q\approx 1$, these arguments will give the rate $N^{-1/d}$ in Theorem \ref{thm:main.hedging.on.Linfty.intro} instead of $N^{-1/2}$.

\subsection{Results for AVaR, OCE, and SF risk measures}
\label{sec:main.avar.etc}

It turns out that for all the specific risk measures discussed in the introduction, the optimal rate $N^{-1/2}$ can be obtained, and with easier arguments.
We therefore state the results for these risk measures separately.
For any measurable $F\colon\mathcal{X}\to\mathbb{R}$, recall that the shortfall risk measure \cite{foellmer01} is defined as
\begin{align*}
 \mathrm{SF}^\mu(F)&:=\inf\Big\{ m\in\mathbb{R} :  E[l(F(S) - m)]\leq 0  \Big\}.
\end{align*}
Here $l\colon\mathbb{R}\to [-1,\infty)$ is a loss function, meaning that $l$ is increasing  and convex such that $1\in\partial l(0)$ (the subdifferential at point 0) and $l(0)=0$.
In other words, $\mathrm{SF}^\mu(F)$ is the smallest capital $m$ by which we should reduce the loss $F$ to make it acceptable, meaning that the expected loss $E[l(F(S)-m)]$ is below the threshold $l(0)=0$.

In a similar spirit, the optimized certainty equivalent (OCE) of Ben-Tal and Teboulle  \cite{ben-tal02,Ben-Teb} is defined via 
\begin{equation}
\label{eq:def.OCE}
	\mathrm{OCE}^\mu(F):=\inf_{m\in\mathbb{R}} ( E[ l(F(S)-m)]+m).
\end{equation}
Again $l$ is a loss function and the interpretation is similar to that of shortfall risk.
Importantly, OCEs cover popular risk measures such as the average value at risk obtained with $l(x)=x^+/(1-u)$ or the entropic risk measure obtained with $l(x) = e^x-1$. 
The following result gives the convergence rate for these particular examples of risk measures:

\begin{theorem}[Rates for AVaR, OCE, and SF]
\label{thm:main.oce.hedging.intro}
	Let $\rho=\mathrm{OCE}$ or $\rho=\mathrm{SF}$ and in the latter case assume that $l$ is strictly increasing.
	There are constants $c,C>0$ such that the following hold.
	\begin{enumerate}[(i)]
	\item
	\label{item:oce.intro.mean}
	We have the moment bound
	\[ E \Big[\big|\pi^{\mu}(F)-\pi^{\mu_N}(F)\big|\Big]
	\leq \frac{C}{\sqrt{N}} \]
	for all $N \geq 1$.
	\item
	\label{item:oce.intro.deviation}
	We have the matching deviation inequality
	\[ P\Big[\big|\pi^{\mu}(F)-\pi^{\mu_N}(F)\big|\geq \varepsilon \Big]
	\leq C\exp\Big(-cN\varepsilon^2\Big)\]
	for all $N\geq 1$ and all $\varepsilon>0$.
	\end{enumerate}
\end{theorem}

The constants $c$ and $C$ depend on $l$, the maximal range of $F$, $G$, the number of options $e$ and the largest Euclidean norm in $\mathcal{G}$.
The rates obtained in both parts of Theorem \ref{thm:main.oce.hedging.intro} are  the usual rates dictated by the central limit theorem and in particular are optimal, see Section \ref{sec:sharpness}.

\begin{remark}
\label{rem:meas.oceN}
	Note that $\pi^{\mu_N}(F)$ is readily checked to be measurable, and we do not need to resort to the outer expectation and probability in Theorem \ref{thm:main.oce.hedging.intro}.
	In fact, if for instance $\rho=\mathrm{OCE}$, then by continuity of the function $l$, we can write
	\[ \pi^{\mu_N}(F):=\inf_{g\in\mathcal{G}, m\in \mathbb{Q}}\frac{1}{N}\sum_{n\le N}l\Big(F(S_n) + \sum_{i=1}^e g_iG_i(S_n) -m\Big) +m.\]
	Recalling that $\mathcal{G}$ is countable, this shows that the random variable $\pi^{\mu_N}(F)$ is measurable.
\end{remark}

As before, Theorem \ref{thm:main.oce.hedging.intro} implies strong consistency:

\begin{corollary}[Consistency]
\label{cor:oce.consistency}
	In the setting of Theorem \ref{thm:main.oce.hedging.intro},  $\lim_{N\to\infty} \pi^{\mu_N}(F) = \pi^{\mu}(F)$ 	$P$-almost surely.
\end{corollary}

\vspace{0.5em}
Let us conclude this subsection with a short discussion on the \emph{biasedness} of $\rho^{\mu_N}(F)$ claimed in the introduction.

\begin{remark}[Biasedness]
\label{rem:bias}
	For typical risk measures,  $\rho^{\mu_N}(F)$ underestimates $\rho(F)$.
	This is easiest explained by considering the optimized certainty equivalents. 
	In fact, we have
	\begin{align*}
	E\big[ \mathrm{OCE}^{\mu_N}(F) \big]
	&=E\Big[ \inf_{m\in\mathbb{R}} \int_{\mathcal{X}} l(F(x)-m)+m \,\mu_N(dx) \Big] \\
	&\leq \inf_{m\in\mathbb{R}} E\Big[ \int_{\mathcal{X}} l(F(x)-m)+m \,\mu_N(dx) \Big] \\  
	&=\mathrm{OCE}^{\mu}(F),
	\end{align*}
	where the last equality follows as $E[ \int f\,d\mu_N]=\int f\,d\mu$ for every $\mu$-integrable function $f$.
	The same applies in the presence of trading, namely $E[ \pi^{\mu_N}(F)]\leq \pi^\mu(F)$.

	More generally, a quick inspection of OCE and SF reveals that both are concave considered as mappings of $\mu$.
	As a matter of fact, this very concavity is the reason for the bias.
	Indeed, concavity and lower--semicontinuity of $\mu\mapsto \rho^\mu(F)$ implies, thanks to Jensen's inequality for infinite dimensional random variables (see e.g. \cite[Theorem 3.1]{Nonnenmacher}), that 
	\[ E[\rho^{\mu_N}(F)]
	\leq \rho^{E[\mu_N]}(F)
	=\rho^\mu(F)\]
	where we used that the measure--valued random variable $\mu_N$ has mean $\mu$.
	While it should be noted that not all law--invariant risk measures are concave in $\mu$, this is often the case\footnote{
	For instance, this is always true for law--invariant comonotonic risk measure, see \cite[Corollary 10]{acciaio2013law}}; see Acciaio and Svindland \cite{acciaio2013law}.
\end{remark}

We further refer to Pitera and Schmidt \cite{Pitera-Sch} for a more in-depth discussion on the issue of biasedness and some empirical evidence.

\subsection{Utility maximization}

It is conceivable that most of the results and methods of the present article extend beyond the estimation of risk measures.
Other issues which seem to fit to our framework and method include the estimation of risk premium principles in insurance, (see e.g.\ Young \cite{Young14} or Furman and Zitikis \cite{Fur-Zit08} for an overview), or estimation of the value of some stochastic optimization problems.

To illustrate the latter, let us consider  another popular approach for quantifying the riskiness of a position, namely utility maximization:
Let $U\colon\mathbb{R}\to\mathbb{R}$ be a concave increasing function and set $u^\mu(F):=E^\mu[U(F(S))]$.
Similar as before, allowing the agent to invest in a market (with stocks returns $G_1,\dots, G_e$), one obtains the utility maximization problem
\[ u_{\mathrm{max}}^\mu(F):=\sup_{g\in\mathcal{G}}  u^\mu\Big( F+\sum_{i=1}^e g_i G_i \Big). \]
In this case, we have the following:

\begin{proposition}[Utility maximization]
\label{prop:utility}
	There are constants $c,C>0$ such that
	\begin{align*}
	E\Big[\big| u_{\mathrm{max}}^\mu(F) - u_{\mathrm{max}}^{\mu_N}(F)\big|\Big]
	&\leq \frac{C}{\sqrt{N}},\\
	P\Big[ \big| u_{\mathrm{max}}^\mu(F) - u_{\mathrm{max}}^{\mu_N}(F)\big|\geq \varepsilon\Big]
	&\leq C\exp\Big(-cN\varepsilon^2\Big)
	\end{align*}
	for all $N\geq 1$ and $\varepsilon>0$.
\end{proposition}

Again note that the rates are optimal and do not depend on the dimension of the underlying nor the number $e$ of available options and that $u_{\mathrm{max}}^{\mu_N}(F)$ is a strongly consistent estimator which typically overestimates its true value (as we deal with maximization instead of minimization this time).

\section{Rates for average value at risk and optimized certainty equivalents}
\label{sec:oce}

Let us briefly fix our notation:
Throughout this paper we make the important convention that $C> 0$ is a generic constant.
This means that $C$ may depend on all kind of parameters (such as some $L^p$ norms of $F$ and $G$, features of the risk measure such as growth of the loss function $l$ in the OCE/SF case,...) but not on $N$.
Moreover, the value of $C$ is allowed to increase from line to line, for instance $\sup_{y}(xy-y^2)=Cx^2\leq Cx^2/2$ or $C\sqrt{e+1}\leq C\sqrt{e}$ for all $e\in\mathbb{N}$, but not $N\leq C$ or $\sqrt{e+1}\leq \sqrt{e}/C$.

To keep the distinction between the analysis for bounded random variables (i.e.\ random variables in $L^p$ for $p=\infty$) and unbounded random variables $(p<\infty$) as light as possible, we shall use the following conventions:
We put $1/\infty:=0$ and for $x>0$, $x^0:=1$ and $x^\infty:=\infty$.

For a metric space $(\Lambda,d_\Lambda)$ and $\varepsilon>0$, denote by $\mathcal{N}(\Lambda,d_\Lambda,\varepsilon)$ the covering numbers at scale $\varepsilon$, that is, $\mathcal{N}(\Lambda,d_\Lambda,\varepsilon)$ is the smallest number for which there is a subset $\tilde{\Lambda}$ with that cardinality satisfying: For every $s\in \Lambda$ there is $\tilde{s}\in \tilde{\Lambda}$ with $d_\Lambda(s,\tilde{s})\leq\varepsilon$.
In other words, $\mathcal{N}(\Lambda,d_\Lambda,\varepsilon)$ the smallest number of balls of radius $\varepsilon$ which covers $\Lambda$.
The latter suggests this to be some measurement of compactness, and in fact, it is an important tool in understanding the behavior of empirical processes, see \cite{van1996weak}.

Recall that $e\in\mathbb{N}$ is a fixed number and $F,G_1,\dots,G_e\colon\mathcal{X}\to\mathbb{R}$ are measurable functions.
For shorthand notation, write $g\cdot G:=\sum_{i=1}^e g_i G_i$ for $g\in\mathcal{G}$ and $|G|:=\sum_{i=1}^e |G_i|$.
Recall that, throughout this article, the set $\mathcal{G}\subset\mathbb{R}^e$ is assumed to be countable and  bounded.
The former assumption is without loss of generality (see Remark \ref{rem:G.countable}) and the latter is shown to be necessary in Proposition \ref{prop:G.needs.to.be.bounded}.

The average value at risk also goes under several different names such as expected shortfall, conditional value at risk, and expected tail loss, and has equally many different (equivalent) definitions, for instances as the value at risk integrated over different levels; see \cite[Section 4.3]{foellmer01} for an overview.
We will only use the definition of AVaR given in \eqref{eq:avar.oce.def}.
Given a loss function $l$, also recall the definition of OCE given in \eqref{eq:def.OCE}.
We additionally assume that $\liminf_{x\to\infty} l(x)/x>1$, which by convexity and $1\in\partial l(0)$ is equivalent to the fact that $l(x)> x$ for some $x\geq 0$.
This assumption is there because $F$ and $G$ are possibly not bounded (in contrast to Section \ref{sec:main.results.}), but not needed if this is the case.

We shall often work under the assumption that $l'$ (the right continuous derivative of the convex function $l$) has polynomial growth of degree $p-1$, which means that $l'(x)\leq C(1+|x|^{p-1})$ for all $x\in\mathbb{R}$.
In particular, recalling the convention $|x|^\infty:=\infty$ for $x\neq 0$, we see that polynomial growth of degree $\infty$ is no restriction at all; for instance, the exponential function $l=\exp$  satisfies this assumption (only) for $p=\infty$.

\vspace{1em}
The goal of this section is to prove Theorem \ref{thm:main.oce.hedging.intro} part \eqref{item:oce.intro.mean}, or rather the following generalization thereof.

\begin{theorem}
\label{thm:main.oce.hedging}
	Let $p\in[1,\infty]$, assume that $l'$ has polynomial growth of degree $p-1$, and that $\|F\|_{L^{2p}(\mu)}$ and $\|G \|_{L^{2p}(\mu)}$ are finite.
	Then
	\[E\Big[\sup_{g\in\mathcal{G}}\big|\mathrm{OCE}^\mu(F+g\cdot G)-\mathrm{OCE}^{\mu_N}(F+g\cdot G)\big|\Big] \leq  \frac{C}{\sqrt{N}} \]
	for all $N\geq 1$.
	The constant $C$ depends on $\mu$ only through the size of the above ${L^{2p}(\mu)}$-norms of $F$ and $G$, on $e$, on $p$, and on the diameter of $\mathcal{G}$.
\end{theorem}

Observe that the measurability of $\mathrm{OCE}^{\mu_N}(F+g\cdot G)$ is readily checked, see for instance Remark \ref{rem:meas.oceN}.

Before presenting the proof of Theorem \ref{thm:main.oce.hedging}, let us shortly elaborate on the  integrability conditions therein.
These assumptions cannot be improved in general.
In fact, consider the trivial case $\mathcal{G} = \{0\}$ and $u=0$ for which we have $\mathrm{AVaR}_0(F(S)) = E[F(S)]$. 
Thus, we have $p=1$ and it is well-known that 
\begin{equation*}
 	E\Big[\Big|\frac{1}{N} \sum_{n\leq N} F(S_n) -E[F(S)] \Big| \Big] 
 	\leq \frac{C}{\sqrt{N}}
\end{equation*}
requires that $F(S)$ has a finite second moment, i.e.\ that  $\|F\|_{L^2(\mu)}$ is finite.\
We now turn to the proof of Theorem \ref{thm:main.oce.hedging}.
In fact, looking at the definition of the optimized certainty equivalent, the reader familiar with the theory of empirical processes recognizes this as a standard problem covered within this theory.
Thus, at some point, an estimate of the covering numbers with respect to the random $L^2(\mu_N)$ norm must be computed.
Fortunately, no geometric arguments are needed, and all randomness can be controlled by some estimates involving moments only.
For this reason it will be useful to keep track of the following quantities
\begin{align}
\label{eq:def.M.MN}
	J:=1+|F|+|G|; \quad M:=\| J\|_{L^p(\mu)} \quad\text{and}\quad M_N:=\| J\|_{L^p(\mu_N)}.
\end{align}
The first result in this spirit is

\begin{lemma}
\label{lem:oce.compact.estimate}
	Assume that $l'$ has polynomial growth of degree $p-1$.
	Then we have that
	\begin{align}
	\label{eq:pi.N.m.bounded}
	|\mathrm{OCE}^\mu(F+g\cdot G)|&\leq C M^p \quad\text{and}\\
	\label{eq:pi.N.m.bounded.1}
	\mathrm{OCE}^\mu(F+g\cdot G)&=\inf_{|m|\leq C M^p} \int_{\mathcal{X}} l(F(x)+g\cdot G(x)-m)+m\,\mu(dx)
	\end{align}
	for every $g\in\mathcal{G}$.
	The same holds true if the pair $(\mu,M)$ is replaced by $(\mu_N,M_N)$ (with the constant $C$ in \eqref{eq:pi.N.m.bounded} not depending on $N$).
\end{lemma}
\begin{proof}
	Assume without loss of generality that $M<\infty$, otherwise there is nothing to show.

	As $l$ is increasing and of polynomial growth with degree $p$ and $\mathcal{G}$ is bounded, we have that
	\begin{equation}
	\label{eq:sup.J}
		\sup_{g\in\mathcal{G}} l(F+g\cdot G)\leq
	\begin{cases}
	C J^p &\text{if } p<\infty,\\
	C  &\text{if } p=\infty.
	\end{cases}
	\end{equation}
	In particular, the choice $m=0$ (in the definition of $\mathrm{OCE}$) and the fact that $l\geq -1$ yield
	\[\mathrm{OCE}^\mu( F+g\cdot G )
	\leq \int_{\mathcal{X}} l(F(x)+g\cdot G(x))\,\mu(dx)
	\leq C M^p \]
	for all $g\in\mathcal{G}$, showing the upper bound in \eqref{eq:pi.N.m.bounded}.
	Further, as $l\geq -1$ and $M\ge1$, this also implies that the infimum over $m$ in the definition of $\mathrm{OCE}^\mu( F+g\cdot G )$ can be restricted to $m\leq C M^p$ for all $g\in\mathcal{G}$.

	On the other hand, by convexity of $l$ and the assumption that $\liminf_{x\to\infty} l(x)/x>1$, there exist $a>1$ and $b\in\mathbb{R}$ such that $l(x)\geq ax-b$ for every $x\in\mathbb{R}$.
	This implies
	\begin{align}
	\label{eq:oce.restrict}
	\begin{split}
	&\int_{\mathcal{X}} l (F(x)+g\cdot G(x)-m)+m \,\mu(dx)\\
	&\geq \int_{\mathcal{X}} a\big(-CJ(x)-m\big) -b +m\,\mu(dx) \\
	&\geq m(1-a) - C M^p,
	\end{split}
	\end{align}
	where we used that $\int_{\mathcal{X}} J\,d\mu\leq M\leq M^p$ which follows from H\"older's inequality and as $M\geq 1$.
	By the previous part we already know that $\mathrm{OCE}^\mu( F+g\cdot G )\leq CM^p$ for all $g\in\mathcal{G}$.
	Together with \eqref{eq:oce.restrict} this implies that the infimum over $m$ in $\mathrm{OCE}^\mu( F+g\cdot G )$ can be restricted to $m\geq -CM^p$ for all $g\in\mathcal{G}$.
	In turn, using once more that $l\geq -1$, this also implies that $\mathrm{OCE}^\mu( F+g\cdot G )\geq -CM^p$ for all $g\in\mathcal{G}$ and thus completes the proof for $(\mu, M)$.

Observe that \eqref{eq:pi.N.m.bounded} and \eqref{eq:pi.N.m.bounded.1} with $(\mu,M)$ replaced by $(\mu_N, M_N)$ is obtained using exactly the same argument as above with $(\mu,M)$ replaced by $(\mu_N, M_N)$.
	In fact, by \eqref{eq:sup.J} we have $\mathrm{OCE}^{\mu_N}(F + g\cdot G) \le CM_N^p$, which implies that the infinum in the definition of $\mathrm{OCE}^{\mu_N}(F + g\cdot G)$ can be restricted to $m\le CM^p_N$ $P$-a.s. for all $g \in \mathcal{G}$.
	On the other hand, as in \eqref{eq:oce.restrict}, we have
	\begin{align*}
	\begin{split}
	&\int_{\mathcal{X}} l (F(x)+g\cdot G(x)-m)+m \,\mu_N(dx)
	\geq m(1-a) - C M^p_N,
	\end{split}
	\end{align*}
	from which we infer that the infimum in the definition of $\mathrm{OCE}^{\mu_N}(F + g\cdot G)$ can be restricted to $m\ge- CM^p_N$ $P$-a.s. for all $g \in \mathcal{G}$.
	This thus shows $\mathrm{OCE}^{\mu_N}(F + g\cdot G)\ge -CM_N^p$.
\end{proof}

\begin{lemma}
\label{lem:covering.numbers.oce}
	Assume that $l'$ has polynomial growth of degree $p-1$, let $m_0\in\mathbb{R}$, and define
	\[ \mathcal{H}:=\Big\{ l (F + g\cdot G -m) + m : g\in\mathcal{G} \text{ and } m\in[-m_0,m_0] \Big\}.  \]
	Then, for every $\varepsilon>0$, we have that
	\[ \mathcal{N}(\mathcal{H},\|\cdot\|_{L^2(\mu_N)},\varepsilon)
	\leq\Big(\frac{C\|J\|^p_{L^{2p}(\mu_N)}}{\varepsilon}\Big)^{e+1}\vee 1 \]
	if $p<\infty$; and $\mathcal{N}(\mathcal{H},\|\cdot\|_{L^2(\mu_N)},\varepsilon)\leq (C/\varepsilon)^{e+1}\vee 1$ if $p=\infty$.
\end{lemma}
\begin{proof}
	Without loss of generality, we work only on the set where $\|J\|_{L^{2p}(\mu_N)}<\infty$ (otherwise there is nothing to show).
	We proceed in two steps.
	\begin{enumerate}[(a)]
	\item Pick two elements $H,\tilde{H}\in\mathcal{H}$ represented as
	\begin{align*}
	H&=l(F+ g\cdot G- m)+m \quad\text{and}   \\
	\tilde{H}&=l(F+\tilde{g}\cdot G-\tilde{m})+\tilde{m}
	\end{align*}
	and define the family of functions $(\varphi_t)_{t\in[0,1]}$ from $\mathcal{X}$ to $\mathbb{R}$ by
	\[\varphi_t:=F+g\cdot G - m + t( ( \tilde{g}-g)\cdot G+m-\tilde{m})\]
	for every $t\in[0,1]$.
	Then $H=l(\varphi_0)+m$ and $\tilde{H}=l(\varphi_1)+\tilde{m}$.
	As $\mathcal{G}$ is bounded, $|\varphi_t|\leq CJ$ for all $t\in[0,1]$.
	By convexity of $l$, its right derivative $l'$ is increasing.
	By the fundamental theorem of calculus, we have
	\begin{align*}
	\| H-\tilde{H}\|_{L^2(\mu_N)}
	&\leq \Big\| \int_0^1 l'(\varphi_t) \partial_t\varphi_t \,dt \Big\|_{L^2(\mu_N)} + |m-\tilde{m}|\\
	&\leq \big\|  l'(CJ ) \big( ( \tilde{g}-g)\cdot G+ m-\tilde{m}\big) \big\|_{L^2(\mu_N)} +|m-\tilde{m}|.
	\end{align*}

	Now note that
	\[\|  l'(CJ ) J\|_{L^2(\mu_N)}
	\leq \begin{cases}
	C\| J \|_{L^{2p}(\mu_N)}^p &\text{if } p<\infty,\\
	C & \text{if } p =\infty.
	\end{cases}\]
	Indeed, for $p<\infty$ this follows from the assumption that $l'(x)\leq C(1+|x|^{p-1})$ for all $x\in\mathbb{R}$, and the fact that $J\geq 1$.
	For $p=\infty$, one has by assumption that $J$ is $\mu$-almost surely bounded.
	Hence, $P$-almost surely, $J$ is also $\mu_N$-almost surely bounded (by the same constant).
	As $l$ is bounded on bounded sets (by convexity), this implies that $l'(J)$ is $\mu_N$-almost surely bounded.

	To conclude, we use once more that $\mathcal{G}$ is bounded and hence $|(\tilde{g}-g)\cdot G|\leq |\tilde{g}-g|J$.
	Therefore
	\begin{align}
	\label{eq:estimate.H.tildeH}
	\| H-\tilde{H}\|_{L^2(\mu_N)}
	&\leq \begin{cases}
	C\| J \|_{L^{2p}(\mu_N)}^p (|g-\tilde{g}|+ |m-\tilde{m}|) &\text{if } p<\infty,\\
	C (|g-\tilde{g}|+ |m-\tilde{m}|) &\text{if } p=\infty.
	\end{cases}
	\end{align}
	In the following we restrict to $p<\infty$ and leave the obvious changes needed when $p = \infty$ to the reader.
	\item
	Fix $\varepsilon>0$ and let $A\subset[-m_0,m_0]$  be such that
	\begin{align*}
	&\text{for all }m\in [-m_0,m_0] \text{ there is } \tilde{m}\in A \text{ with }|m-\tilde{m}|\leq \frac{\varepsilon}{2C\| J \|_{L^{2p}(\mu_N)}^p}
	\end{align*}
	and $B\subset\mathcal{G}$ such that
	\begin{align*}
	&\text{for all }g\in\mathcal{G} \text{ there is } \tilde{g}\in B \text{ with }|g-\tilde{g}|\leq \frac{\varepsilon}{2C\| J \|_{L^{2p}(\mu_N)}^p}.
	\end{align*}
	Then, if we define $\tilde{\mathcal{H}}$ exactly as $\mathcal{H}$ only with $[-m_0,m_0]$ replaced by $A$ and $\mathcal{G}$ replaced by $B$, then by \eqref{eq:estimate.H.tildeH}, for every $H\in\mathcal{H}$ there is $\tilde{H}\in\tilde{\mathcal{H}}$ with $\| H-\tilde{H}\|_{L^2(\mu_N)} \leq\varepsilon$.

	This implies that
	\begin{align*}
	\mathcal{N}(\mathcal{H},\|\cdot\|_{L^2(\mu_N)},\varepsilon)
	&\leq \mathop{\mathrm{card}}( \tilde{\mathcal{H}})\\
	&\leq \mathop{\mathrm{card}}( A\times B )
	=\mathop{\mathrm{card}}(A) \mathop{\mathrm{card}}(B)
	\end{align*}
	where $\mathop{\mathrm{card}}$ means cardinality.

	The set $A$ can be constructed simply by an equidistant partition of $[-m_0,m_0]$ at cardinality $\mathop{\mathrm{card}}(A)\leq (C\| J \|_{L^{2p}(\mu_N)}^p/\varepsilon)\vee 1$.
	In a similar manner, $B$ can be constructed with $\mathop{\mathrm{card}}(B)\leq (C\| J \|_{L^{2p}(\mu_N)}^p/\varepsilon)^e\vee 1$.
	\end{enumerate}
	Combining both steps yields the proof.
\end{proof}

In order to apply results from theory of empirical processes, we need the following observation.

\begin{lemma}
\label{lem:pw.meausrable}
	The set $\mathcal{H}$ defined in Lemma \ref{lem:covering.numbers.oce} satisfies Assumption \ref{ass:pw.mb}.
\end{lemma}
\begin{proof}
	Let $A\subset[-m_0,m_0]$ be countable and dense and define
	\[ \mathcal{H}':=\Big\{ l(F+ g\cdot G-m) : g\in \mathcal{G} \text{ and } m\in A\Big\}.\]
	The set $\mathcal{H}'$ is clearly countable.
	Let $H=l(F+g\cdot G-m)\in\mathcal{H}$ and let $(m^n)_n\subset A$ be a sequence which converges to $m$.
	Then $H^n:=l(F+g\cdot G-m^n)\in\mathcal{H}'$ converges pointwise to $H$ and in $L^2(\nu)$ for every measure $\nu$ such that $l(|F|+|G|)\in L^2(\nu)$  (by dominated convergence).
\end{proof}

Inspecting the proof actually yields the following result, which we state for later reference.

\begin{corollary}
\label{cor:covering.compact}
	Let $m_0\in\mathbb{R}$, let $f\colon\mathbb{R}\to\mathbb{R}$ be locally Lipschitz continuous, and assume that $J$ is bounded.
	Then it holds that\footnote{Observe that $\|f\|_\infty:= \sup_{x\in \mathbb{R}}|f(x)|$ represents the supremum norm of $f$; and not the essential supremum norm that we denote $\|\cdot\|_{L^\infty}$.}
	\[ \mathcal{N}\Big(\big\{ f(F + g\cdot G -m): g\in\mathcal{G} \text{ and } m\in[-m_0,m_0] \big\}, \|\cdot\|_{\infty},\varepsilon\Big)
	\leq \Big(\frac{C}{\varepsilon}\Big)^{e+1}\vee 1\]
	for every $\varepsilon>0$.
\end{corollary}

We are now ready for the
\begin{proof}[Proof of Theorem \ref{thm:main.oce.hedging}]
	For shorthand notation, set
	\[ \Delta_N:=\sup_{g\in\mathcal{G}}\Big|\mathrm{OCE}^\mu(F+g\cdot G)-\mathrm{OCE}^{\mu_N}(F+g\cdot G)\Big|\]
	for every $N\geq 1$.
	With $M$ and $M_N$ defined in \eqref{eq:def.M.MN}, we write
	\begin{align*}
	E[\Delta_N]
	=E[\Delta_N1_{M_N\leq M+1}] + E[\Delta_N1_{M_N>M+1}]
	\end{align*}
	and investigate both terms separately.

	\begin{enumerate}[(a)]
	\item
	We start with the first term.
	Lemma \ref{lem:oce.compact.estimate} guarantees that
	\[ \Delta_N 1_{M_N\leq M+1}\leq \sup_{H\in\mathcal{H}} \Big|\int_{\mathcal{X}} H(x)\,(\mu-\mu_N)(dx)\Big| \]
	for every $N\geq 1$, where
	\[\mathcal{H}:=\{ l (F + g\cdot G -m) + m : g\in\mathcal{G} \text{ and } |m|\leq C(M+1)^p \}.\]
	By Lemma \ref{lem:pw.meausrable}, the set $\mathcal{H}$ satisfies Assumption \ref{ass:pw.mb}.
	Therefore, the `empirical process version'  of Dudley's entropy-integral theorem (i.e.\ Theorem \ref{thm:dudley}).
	implies that
	\begin{align*}
	&E\Big[\sup_{H\in\mathcal{H}} \Big|\int_{\mathcal{X}} H(x)\,(\mu-\mu_N)(dx)\Big| \Big]\\
	&\leq \frac{C}{\sqrt{N}} \Big( E[ H^\ast(S)^2]^{\frac{1}{2}} + E\Big[ \int_0^\infty \sqrt{ \log \mathcal{N}(\mathcal{H},\|\cdot\|_{L^2(\mu_N)},\varepsilon) } \,d\varepsilon\Big] \Big)
	\end{align*}
	for all $N\geq1$, where $H^\ast:=l(F+g^\ast\cdot G)\in\mathcal{H}$ for some $g^\ast\in\mathcal{G}$.
	By definition of $J$ we have that $E[ H^\ast(S)^2]^{\frac{1}{2}}\leq C \|J\|_{L^{2p}(\mu)}^p$.
	It remains to gain control over the entropy integral term.
 
	Assume first that $p<\infty$.
	Then, estimating the covering numbers of $\mathcal{H}$ by means of Lemma \ref{lem:covering.numbers.oce} implies that
	\begin{align*}
	& E \Big[ \int_0^\infty \sqrt{ \log \mathcal{N}(\mathcal{H},\|\cdot\|_{L^2(\mu_N)},\varepsilon) } \,d\varepsilon\Big]\\
	&\leq  C E \Big[ \int_0^\infty \sqrt{ \log \big(\tfrac{C\|J\|_{L^{2p}(\mu_N)}^p}{\varepsilon}\vee 1\big) } \,d\varepsilon\Big]\\
	&\le C E \Big[ \|J\|_{L^{2p}(\mu_N)}^p \int_0^\infty \sqrt{ \log \big(\tfrac{1}{\tilde{\varepsilon}}\vee 1\big) } \,d\tilde{\varepsilon}\Big],
	\end{align*}
	where the last inequality follows from substituting $\varepsilon$ by  $\tilde{\varepsilon}:=\varepsilon/ C\|J\|_{L^{2p}(\mu_N)}^p$.
	In a final step, notice that
	\[\int_0^\infty \sqrt{ \log \big(\tfrac{1}{\varepsilon}\vee 1\big) } \,d\varepsilon<\infty
	\quad\text{and}\quad
	E [\|J\|_{L^{2p}(\mu_N)}^p]\leq C \|J\|_{L^{2p}(\mu)}^p. \]
	The second statement follows from Jensen's inequality.
	Therefore
	\[ E [\Delta_N1_{M_N\leq M+1}]\leq \frac{C}{\sqrt{N}}\]
	for all $N\geq 1$, showing that the first term behaves as required.
	If $p=\infty$ the same arguments apply (with $\|J\|_{L^{2p}(\mu_N)}^p$ replaced by a constant and Corollary \ref{cor:covering.compact} applied instead of Lemma \ref{lem:covering.numbers.oce}) and we again obtain $E [\Delta_N1_{M_N\leq M+1}]\leq C/ \sqrt{N}$.
	\item
	As for the second term, applying H\"older's inequality yields
	\begin{align}
	\label{eq:second.term.moments}
	E [\Delta_N 1_{M_N>M+1}]
	\leq E [\Delta_N^2]^{1/2} P[ M_N>M+1 ]^{1/2}.
	\end{align}
	We start by estimating $P[M_N>M+1]^{1/2}$.
	For $p=\infty$, one has $P[M_N>M+1]=0$ for all $N$.
	For $p<\infty$, using first that $M,M_N\geq 1$ and then Chebycheff's inequality, we estimate
	\begin{align*}
	P [ M_N-M > 1]
	&\leq P [M_N^p-M^p > 1] \\
	&\leq E [ (M_N^p-M^p)^2].
	\end{align*}
	Further, making use of the fact that the $(S_1,\dots,S_N)$ are independent with   $M^p=E[J(S_n)^p]$ for all $n$, one has
	\begin{align*}
	E [ (M_N^p-M^p)^2]
	&=E \Big[ \Big(\frac{1}{N}\sum_{n=1}^N \big(  J(S_n)^p-E[J(S_n)^p] \big) \Big)^2\Big]\\
	&=\frac{1}{N} E [(J(S_1)^p-E[J(S_1)^p])^2]\\
	&\leq \frac{2\|J\|_{L^{2p}(\mu)}^{2p}}{N}.
	\end{align*}
	This shows that $P[M_N>M+1]^{1/2}\leq C /\sqrt{N}$.
	
	Regarding $E [\Delta_N^2]$, use Lemma \ref{lem:oce.compact.estimate} to estimate
	\begin{align*}
	E [\Delta_N^2]
	&\leq C(M^{2p} +  E[ M_N^{2p}]).
	\end{align*}	
	The same arguments as above show that  $E[ M_N^{2p}]\leq \|J\|_{L^{2p}(\mu)}^{2p}$.
	Plugging both estimates in \eqref{eq:second.term.moments} shows that
	\[E[\Delta_N 1_{M_N>M+1}]\leq \frac{C}{\sqrt{N}}\]
	for all $N\geq 1$.
	\end{enumerate}
	Putting both estimates together, we obtain $E[\Delta_N]\leq C/\sqrt{N}$ for all $N\geq 1$.
	This completes the proof.
\end{proof}

\section{General law--invariant risk measures}
\label{sec:general.rm}

This section deals with general risk measures, which we start by briefly describing.
First, in order to allow for unbounded $F$ and $G$, one needs to define risk measures for unbounded functions.
A function $\rho\colon L^p\to\mathbb{R}$ with $p\in[1,\infty]$ is again called a (convex) law--invariant risk measure if (a)-(d) of Definition \ref{def:risk.measure.l.infty} hold with $L^\infty$ replaced by $L^p$.
Further recall that $\rho$ is called coherent if in addition $\rho(\lambda X )=\lambda\rho(X)$ for all $X\in L^p$ and $\lambda\geq 0$.

As already mentioned, by \cite{Jou-Sch-Tou}, every law--invariant risk measure automatically satisfies the Fatou-property
as well as the  \emph{spectral representation}\footnote{
	It also goes under the name Kusuoka representation as the $L^\infty$-version was discovered by Kusuoka \cite{kusuoka2001law}. }
\begin{align}
\label{eq:rho.spectral.rep}
\rho(X)=\sup_{\gamma\in\mathcal{M}} \Big( \int_{[0,1)} \mathrm{AVaR}_u(X)\,\gamma(du) - \beta(\gamma) \Big)
\quad\text{for } X\in L^p.
\end{align}
See \cite{Gao2018} for the case of unbounded random variables.
Here $\mathcal{M}$ is a subset of probability measures on $[0,1)$ armed with its Borel $\sigma$-field, $\beta\colon\mathcal{M}\to[0,\infty)$ is a convex function, and $\mathrm{AVaR}$ is the average value at risk, defined in \eqref{eq:avar.oce.def}.
Note that $\mathrm{AVaR}$ is evidently a coherent law--invariant risk measure.
Recall the definition of $J:= 1+|F|+|G|$ already given in \eqref{eq:def.M.MN}.

Before we are ready to state the generalization of part \eqref{item:oce.intro.mean} of Theorem \ref{thm:main.hedging.on.Linfty.intro}, the treatment of unbounded $F,G$ requires one last definition:
for every parameter $p\in[1,\infty]$ and $x\geq 0$ set
\[w_p(x):=\sup\{ \rho(X) : \|X\|_{L^p}\leq x \}.\]
Note that $w_p$ is convex, nonnegative, and $w_p$ grows at least linearly.
Moreover, in the important case of a coherent risk measure or if $p=\infty$, the function $w_p$ is linear.

\begin{theorem}
\label{thm:main.hedging.on.lp}
	Let $1<q\leq p\leq \infty$ and let $\rho\colon L^p\to\mathbb{R}$ be a law--invariant risk measure which is $q$-regular.
	Assume that $\mathcal{G}$ is bounded and that 
	\begin{enumerate}[(a)]
	\item $\| w_p(t J^p) \|_{L^2(\mu)}<\infty$ for every $t\geq 0$ in case that $p<\infty$,
	\item $J$ is bounded in case that $p=\infty$.
	\end{enumerate}
	Then
	\[E\Big[\sup_{g\in\mathcal{G}}\big| \rho^\mu(F+g\cdot G)-\rho^{\mu_N}(F+g\cdot G)\big|\Big]
	\leq \frac{C}{N^\frac{1/q-1/2p}{2-1/p}} \]
	for all $N\geq 1$.
\end{theorem}

Note that if $w_p$ is linear, then $\|w_p(t J^p) \|_{L^2(\mu)}<\infty$ simply means that $\|J\|_{L^{2p}(\mu)}$ is finite.
In general, $\|w_p(t J^p) \|_{L^2(\mu)}<\infty$ always implies that $\|J\|_{L^{2p}(\mu)}<\infty$ (by convexity of $w_p$).  
For convenience, we compute in Table 1 some values of the convergence rates obtained in Theorem \ref{thm:main.hedging.on.lp}: 

\begin{figure}[ht]
\centering
\begin{tabular}{ c|c|c|c|c|c }
  & $q\approx 1$
  & $q=2$
  & $q= p $
  & $p=\infty$
  \\ \hline
 	$R_{p,q}:=\frac{1/q-1/2p}{2-1/p}$
  & $\approx\frac{1}{2}$
  & $\frac{p-1}{4p-2}$
  & $\frac{1}{2(2p-1)}$
  & $\frac{1}{2q}$
\end{tabular}\\
\vspace{.3cm}
Table 1: Convergence rates for different values of $p$ and $q$.
\label{tab:rates}
\end{figure}
Further observe that the rate $R_{p,q}$ is increasing in $p$ and decreasing in $q$.
The idea of the proof of Theorem \ref{thm:main.hedging.on.lp} is the following:
By Section \ref{sec:oce} we understand the behavior of the mean error for the average value at risk (being a special case of the optimized certainty equivalents).
By the spectral representation \eqref{eq:rho.spectral.rep}, AVaR forms the building block of every law--invariant risk measure and we conclude via a (multiscale) approximation, keeping track of the risk aversion parameter $u$ of the average value at risk (which will make all constants explode when approaching $u\approx 1$) and the growth of measures $\gamma(du)$ in the spectral representation \eqref{eq:rho.spectral.rep} (which only puts little mass on $u\approx 1$).

\vspace{0.5em}
The preparatory work needed is done in the next few lemmas.
\begin{lemma}
\label{lem:avar.polynomial.density}
	Let the assumptions of Theorem \ref{thm:main.hedging.on.lp} be satisfied.
	Let $X^\ast$ be Pareto distributed with scale parameter 1 and shape parameter $q$.
	Then we have that
	\[\mathrm{AVaR}_u(X^\ast)
	=\frac{q}{q-1}\frac{1}{(1-u)^{1/q}}\]
	for every $u\in[0,1)$.
\end{lemma}
\begin{proof}
	The proof follows from an elementary calculation, e.g.\ by involving the quantile representation 
	\[ \mathrm{AVaR}_u(X^\ast)=\frac{1}{1-u}\int_{u}^{1} q_{X^\ast}(t)\, dt \]
	of the average value at risk \cite[Proposition 4.51]{foellmer01}, where $q_{X^\ast}(t)$ denotes the $t$-quantile of $X^\ast$.
\end{proof}

\begin{lemma}
\label{lem:avar.moments}
Let the assumptions of Theorem \ref{thm:main.hedging.on.lp} be satisfied.
	For every $p\in(1,\infty]$ and $X\in L^p$ we have that
	\[ |\mathrm{AVaR}_u(X)| \leq \frac{\|X\|_{L^p}}{(1-u)^{1/p}}\]
	for every $u\in[0,1)$. 
\end{lemma}
\begin{proof}
	For $p=\infty$ the claim is trivial.
	For $p<\infty$ we again involve the quantile representation of the average value at risk and apply H\"older's inequality
	\begin{align*}
	| \mathrm{AVaR}_u(X) |
	&= \frac{1}{1-u} \Big| \int_{[0,1]} 1_{[u,1]}(t) q_{X}(t)\, dt \Big|
	\leq \frac{ (1-u)^{\frac{p-1}{p}} }{1-u} \Big( \int_{[0,1]}  q_{X}(t)^p \, dt \Big)^{\frac{1}{p}}.
	\end{align*}
	As the last integral equals $\|X\|_{L^p}$, this completes the proof.
\end{proof}

\begin{lemma}
Let the assumptions of Theorem \ref{thm:main.hedging.on.lp} be satisfied.
	For every fixed $a>0$, there exists a constant $b>0$ such that	
	\[ \rho(X)=\sup_{\gamma\in\mathcal{M} : \text{ s.t.\ } \beta(\gamma)\leq b} \Big( \int_{[0,1)} \mathrm{AVaR}_u(X)\,\gamma(du) - \beta(\gamma) \Big)  \]
	for all $X\in L^p$ with $\|X\|_{L^p}\leq a$.
\end{lemma}
\begin{proof}
	Let $X^\ast$ be the random variable of Lemma \ref{lem:avar.polynomial.density}.
	\begin{enumerate}[(a)]
	\item
	In a first step we show that $|\rho(X)|\leq C$ for all $X\in L^p$ with $\|X\|_{L^p}\leq a$.
	For such $X$, by Lemma \ref{lem:avar.polynomial.density} and Lemma \ref{lem:avar.moments}, one has that
	\begin{align}
	\label{eq:avar.X.dominated.by.Xast}
	\mathrm{AVaR}_u(|X|)
	\leq \frac{a}{(1-u)^{1/p}}
	\leq \frac{a}{(1-u)^{1/q}}
	= \mathrm{AVaR}_u(C X^\ast )
	\end{align}
	for every $u\in[0,1)$.
	Here we used that $q\leq p$, hence $(1-u)^{1/p}\geq (1-u)^{1/q}$.
	Therefore
	\begin{align*}
	\rho(|X|)
	&\leq \sup_{\gamma\in\mathcal{M}} \Big( \int_{[0,1)} \mathrm{AVaR}_u(|X|)\,\gamma(du) - \beta(\gamma) \Big)\\
	&\leq \sup_{\gamma\in\mathcal{M}} \Big( \int_{[0,1)} \mathrm{AVaR}_u(C X^\ast)\,\gamma(du) - \beta(\gamma) \Big)\\
	&\leq \sup_{n\in \mathbb{N}}\sup_{\gamma\in\mathcal{M}} \Big( \int_{[0,1)} \mathrm{AVaR}_u(C X^\ast\wedge n)\,\gamma(du) - \beta(\gamma) \Big)
	=\sup_{n\in \mathbb{N}} \rho(CX^\ast \wedge n)
	\end{align*}
	for every $X$ with $\|X\|_{L^p}\leq a$, where the latter inequality follows by monotone convergence.
	Note that $CX^\ast$ again follows a Pareto distribution with shape parameter $q$ and hence $\sup_{n\in \mathbb{N}}\rho(CX^\ast \wedge n)$ is finite by definition of $q$-regularity.

	It further follows by convexity and monotonicity of $\rho$ together with $\rho(0)=0$, that $|\rho(X)|\leq \rho(|X|)$ for all $X \in L^p$.
	This implies that indeed $|\rho(X)|\leq C$ for all $X\in L^p$ with $\|X\|_{L^p}\leq a$.
	\item
	We proceed to prove the claim.
	Define
	\[ \varphi\colon\mathbb{R}_+\to[0,\infty]\quad \text{by} \quad \varphi(y):=\sup_{x\in\mathbb{R}_+} \Big( xy -  \sup_{n\in\mathbb{N} }\rho(xX^\ast\wedge n)\Big). \]
	Then $\varphi$ is convex, increasing, and as $ \sup_{n\in\mathbb{N}} \rho(xX^\ast\wedge n)<\infty$ for all $x\in\mathbb{R}_+$, one can verify that $\varphi(y)/y\to\infty$ as $y\to\infty$.
	Now note that the (spectral) representation of $\rho$ in \eqref{eq:rho.spectral.rep} implies that
	\[  \sup_{n\in\mathbb{N}} \rho(xX^\ast\wedge n) \geq \int_{[0,1)} \mathrm{AVaR}_u(xX^\ast)\,\gamma(du) - \beta(\gamma) \]
	for all $x\geq 0$  and $\gamma\in\mathcal{M}$.
	Therefore, one has
	\begin{align*}
	\beta(\gamma)
	&\geq \sup_{x\geq 0} \Big( \int_{[0,1)} \mathrm{AVaR}_u(xX^\ast)\,\gamma(du) - \rho(x X^\ast) \Big)\\
	&=\varphi\Big( \int_{[0,1)} \mathrm{AVaR}_u(X^\ast)\,\gamma(du)\Big)
	\end{align*}
	for every $\gamma\in\mathcal{M}$.
	For every $X$ with $\|X\|_{L^p}\leq a$, by \eqref{eq:avar.X.dominated.by.Xast} one has
	\begin{align}
	\label{eq:avar.beta.leq.b}
	\begin{split}
	\int_{[0,1)} \mathrm{AVaR}_u(X)\,\gamma(du) - \beta(\gamma)
	&\leq C \int_{[0,1)} \mathrm{AVaR}_u(X^\ast)\,\gamma(du) - \beta(\gamma) \\
	&\leq C \varphi^{-1}(\beta(\gamma))-\beta(\gamma),
	\end{split}
	\end{align}
	where $\varphi^{-1}$ denotes the (right)-inverse of $\varphi$.

	As $\varphi(y)/y\to\infty$ when $y\to\infty$, one has that $\varphi^{-1}(x)/x\to 0$ when $x\to \infty$ which implies that
	\begin{align}
	\label{eq:varphi.growth}
	C \varphi^{-1}(\beta(\gamma))-\beta(\gamma)\to-\infty
	\quad\text{when } \beta(\gamma)\to\infty.
	\end{align}

	Now recall that $\rho(X)$ equals the supremum over $\gamma\in\mathcal{M}$ of the left hand side of \eqref{eq:avar.beta.leq.b} and that $|\rho(X)|\leq C$ for all $X$ with $\|X\|_{L^q}\leq a$ by the first part of this proof.
	Therefore \eqref{eq:varphi.growth} implies that there is some constant $b$ such that only $\gamma\in\mathcal{M}$ for which $\beta(\gamma)\leq b$ need to be considered in the computation of $\rho(X)$.
	\qedhere
	\end{enumerate}
\end{proof}

\begin{lemma}
\label{lem:growth.Gamma}
Let the assumptions of Theorem \ref{thm:main.hedging.on.lp} be satisfied.
	For every fixed $b\in\mathbb{R}_+$, we have
	\[ \Gamma_b([r,1))
	:=\sup_{\gamma\in\mathcal{M} \text{ s.t.\ }\beta(\gamma)\leq b} \gamma([r,1))
	\leq C (1-r)^{1/q}  \]
	for every $r\in[0,1)$.
\end{lemma}
\begin{proof}
	Let $X^\ast$ be the random variable of Lemma \ref{lem:avar.polynomial.density}.
	Then it follows from interchanging two suprema in the spectral representation \eqref{eq:rho.spectral.rep} (one over $n$ and one over $\gamma$), monotone convergence (applied under each $\gamma$), and Lemma \ref{lem:avar.polynomial.density} that
	\begin{align}
	\nonumber \sup_n\rho(X^\ast \wedge n)
	&=\sup_{\gamma\in \mathcal{M}} \sup_n\Big( \int_{[0,1)} \mathrm{AVaR}_u(X^\ast \wedge n)\,\gamma(du)-\beta(\gamma)  \Big)\\
	\label{eq:kusuoka.plynomial.density}
	&\geq \sup_{\gamma\in \mathcal{M}\text{ s.t.\ }\beta(\gamma)\leq b} \Big( \int_{[0,1)} \frac{q}{q-1} \frac{1}{(1-u)^{1/q}}\,\gamma(du)-\beta(\gamma)  \Big).
	\end{align}
	By assumption $\sup_n\rho(X^\ast\wedge n)\in\mathbb{R}$, which implies that
	\begin{align*}
	&\sup_{\gamma\in \mathcal{M}\text{ s.t.\ }\beta(\gamma)\leq b}  \int_{[0,1)}\frac{1}{(1-u)^{1/q}}\,\gamma(du) \\
	&\leq \frac{q-1}{q} \Big(\sup_n\rho(X^\ast\wedge n) + b +1 \Big)
	=C.
	\end{align*}
	In particular, this implies that 
	\begin{align*}
	&\Gamma_b([r,1))
	\leq \sup_{\gamma\in \mathcal{M}\text{ s.t.\ }\beta(\gamma)\leq b}  \int_{[r,1)} \Big( \frac{1-r}{1-u}\Big)^{1/q}\,\gamma(du)
	\leq C (1-r)^{1/q},
	\end{align*}
	which proves the claim.
\end{proof}

\begin{lemma}
\label{lem:strange.sums.general}
	Let the assumptions of Theorem \ref{thm:main.hedging.on.lp} be satisfied.
	Let $0\leq b<a<1$.
	Then it holds that
	\[\sum_{n\geq 1}  2^{-an} \cdot \Big( (x2^n) \wedge 2^{bn} \Big)
	\leq C\Big( x^{\frac{a-b}{1-b}} \vee x \Big) 	\]
	for every $x\in[0,\infty)$ (where $C$ does not depend on $x$).
\end{lemma}
\begin{proof}
	For $x=0$ there is nothing to prove.
	We now consider the case $x\in(0,1]$, denote by $s_n$ the summand, and set
	\[n_N:=\frac{\log(1/x)}{(1-b)\log 2}.\]
	Then a quick computation reveals
	\[	s_n
	=\begin{cases}
	x 2^{n(1-a)}  & \text{if } n< n_N, \\
	2^{n(b-a)}  & \text{if } n\geq n_N.
	\end{cases}
	\]

	By properties of the geometric series one has
	\begin{align*}
	\sum_{n< n_N} s_n
	&= C x \sum_{n < n_N} 2^{n(1-a)}\\
	&\leq C x \frac{1-2^{n_N(1-a)}}{1-2^{1-a}}
	\leq C x 2^{n_N(1-a)}.
	\end{align*}
	Moreover, as $s^{\log t}=t^{\log s}$ for $s,t>0$, the definition of $n_N$ implies that
	\begin{align}
	\label{eq:strange.log}
	\begin{split}
	2^{n_N(1-a)}
	&=\Big( 2^{\frac{1-a}{(1-b)\log 2}}\Big)^{\log(1/x)} \\
	&=\Big(\frac{1}{x}\Big)^{\log\big(2^{\frac{1-a}{(1-b)\log 2}}\big)}
	=x^{\frac{a-1}{1-b}}.
	\end{split}
	\end{align}
	Putting everything together, this implies
	\begin{align*}
	\sum_{n< n_N} s_n
	\leq C x \cdot x^{\frac{a-1}{1-b}}
	=C x^{\frac{a-b}{1-b}}.
 	\end{align*}

	For the tail of the sum, the same computation as in \eqref{eq:strange.log} shows that $2^{n_N(b-a)}	=x^{\frac{b-a}{1-b}}$.
	Therefore, another application of the geometric series properties implies that
	\begin{align*}
	\sum_{n\geq n_N} s_n
	&=\sum_{n\geq n_N} 2^{n(b-a)}\\
	&\leq \frac{2^{n_N(b-a)}}{1-2^{b-a}}
	\leq C 2^{n_N(b-a)}
	=C x^{\frac{a-b}{1-b}}.
	\end{align*}
	Hence, adding the sums over $n< n_N$ and $n\geq n_N$ and noting that $(a-b)/(1-b)\in(0,1)$ and hence $x\leq x^{(a-b)/(1-b)}$ for $x\in[0,1]$ yields the claim for $x\in(0,1]$.

	For $x\geq 1$ we have $x\geq  x^{(a-b)/(1-b)}$ and
	\[\sum_{n\geq 1}  2^{-an} \cdot \Big( (x2^n) \wedge 2^{bn} \Big)
	\leq x \sum_{n\geq 1}  2^{-an} \cdot \Big( 2^n \wedge 2^{bn} \Big)
	\leq Cx,\]
	where the last inequality follows from convergence of the geometric series / the previous step.
\end{proof}

For every $N\geq 1$ and $u\in[0,1)$ define
\begin{align}
\label{eq:def.delta.u}
\delta_u^N := \sup_{g\in\mathcal{G}} \Big| \mathrm{AVaR}_u^\mu(F+ g\cdot G ) - \mathrm{AVaR}_u^{\mu_N}(F+ g\cdot G )\Big|.
\end{align}
The following lemma controls uniformly the behavior of $\delta$.
Observe that measurability of $\mathrm{AVaR}_u^{\mu_N}(F+ g\cdot G)$ was addressed in Remark \ref{rem:meas.oceN}.
In particular, $\delta^N_u$ is measurable for every $u\in[0,1)$, and a quick argument shows that $\sup_{u\in[0,v]} \delta^N_u$ remains measurable for every $v\in[0,1)$.

\begin{lemma}
\label{lem:control.delta.general.risk.lp}
	Let the assumptions of Theorem \ref{thm:main.hedging.on.lp} be satisfied.
	We have that 
	\[ E\Big[ \sup_{u\in[0,v]} \delta_u^N \Big]
	\leq  \frac{C}{(1-v)\sqrt{N}} \wedge \frac{C}{(1-v)^{1/2p}} \]
	for every $v\in(0,1)$.
\end{lemma}
\begin{proof}
	We start with the easier estimate, namely that
	\begin{align}
	\label{eq:estimate.delta.u.general.risk.simple}
	E\Big[ \sup_{u\in[0,v]} \delta_u^N \Big]
	\leq   \frac{ C }{ (1-v)^{1/2p} }.
	\end{align}
	 As $|F+g\cdot G|\leq CJ$ for every $g\in\mathcal{G}$, monotonicity of $\mathrm{AVaR}_u$ implies $\mathrm{AVaR}_u^\mu(F+ g\cdot G ) \leq \mathrm{AVaR}_u^\mu(CJ)$ for every $g\in\mathcal{G}$; similarly with $\mu$ replaced by $\mu_N$.
	Now Lemma \ref{lem:avar.moments} implies
	\[  \sup_{u\in[0,v]}\delta_u^N
	\leq  \frac{\|CJ\|_{L^{2p}(\mu)} + \|CJ\|_{L^{2p}(\mu_N)} }{(1-v)^{1/2p}}. \]
	Further Jensen's inequality implies $E[\|CJ\|_{L^{2p}(\mu_N)}]\leq \|CJ\|_{L^{2p}(\mu)}$ and thus we get \eqref{eq:estimate.delta.u.general.risk.simple}.

	To conclude the proof, we are left to prove that
	\begin{align}
	\label{eq:estimate.delta.u.general.risk}
	E\Big[ \sup_{u\in[0,v]}\delta_u^N\Big]
	\leq \frac{C}{(1-v)\sqrt{N}},
	\end{align}
	which we shall do in several steps.
	\begin{enumerate}[(a)]
	\item
	Define
	\[ \mathcal{H}:=\{\varphi(F+g\cdot G) : \varphi\colon\mathbb{R}\to\mathbb{R} \text{ is 1-Lipschitz, }\varphi(0)=0  \text{ and } g\in\mathcal{G}\}. \]
	Then it holds that
	\begin{align}
	\label{eq:delta.v.lipschitz}
	\sup_{u\in[0,v]} \delta_u^N
	&\leq \frac{1}{1-v}\sup_{H\in\mathcal{H}}\Big| \int_{\mathcal{X}} H\,(\mu-\mu_N)(dx)\Big|.
	\end{align}
	Indeed, every function appearing in the definition of $\mathrm{AVaR}_u$ is of the form 	$\varphi(F+g\cdot G)/(1-u)$ for a $1$-Lipschitz function, see \eqref{eq:avar.oce.def}.
	Subtracting $\varphi(F(0)+g\cdot G(0))/(1-u)$ does not change the value of the difference of two integrals, which yields the claim.
		\item
	We proceed to compute the covering numbers of $\mathcal{H}$.
	First observe that since $\mathcal{G}$ is bounded, there is a constant $C_0$ such that $|F+ g\cdot G|\leq C_0 J$ and  $|(g-\tilde{g})\cdot G|\leq C_0 J$ for all $g,\tilde{g}\in\mathcal{G}\cup\{0\}$.
	The value of $C_0$ will be kept fixed throughout this proof.
	Let $\varepsilon>0$ and set 
	\begin{align}
	\label{eq:cover.a.eps}
	a_\varepsilon:=
	\begin{cases}
	\frac{(6C_0)^{1/p} \| J \|_{L^{2p}(\mu_N)}}{\varepsilon^{1/p}}, &\text{if }p<\infty,\\
	\| J \|_{L^{\infty}(\mu_N)}, &\text{if }p=\infty.
	\end{cases}
	\end{align}
	
	First, let $\tilde{L}_\varepsilon$ be a set of 1-Lipschitz functions from $\mathbb{R}$ to $\mathbb{R}$ which vanish at zero such that for every $1$-Lipschitz function $\varphi$ there is $\tilde{\varphi}\in \tilde{L}_\varepsilon$ satisfying  $\sup_{t\in[-C_0a_\varepsilon,C_0a_\varepsilon]} |\varphi(t)-\tilde{\varphi}(t)|\leq \varepsilon/3$.
	Such a set $\tilde{L}_\varepsilon$ can be constructed with
	\begin{align}
		\label{eq:covering.card.L}
		\mathrm{card}(\tilde{L}_\varepsilon)
	\leq \exp\Big( \frac{C}{(\varepsilon/a_\varepsilon)\wedge 1} \Big);
	\end{align}
	we detail this in step (c) below.
	Moreover, let $\tilde{\mathcal{G}}_\varepsilon\subset\mathcal{G}$ be such that for every $g\in\mathcal{G}$ there is $\tilde{g}\in\tilde{\mathcal{G}}_\varepsilon$ satisfying $|g-\tilde{g}|\leq \varepsilon/(3C_0a_\varepsilon)$.
	Such a set $\tilde{\mathcal{G}}_\varepsilon$ can be constructed with
	\begin{align}
		\label{eq:covering.card.G}
		\mathrm{card}(\tilde{\mathcal{G}}_\varepsilon)	\leq \Big( \frac{C}{(\varepsilon/a_\varepsilon)\wedge 1} \Big)^e,
	\end{align}
	using an equidistant grid of the bounded set $\mathcal{G}\subset\mathbb{R}^e$.
	
	Now set 
	\[ \tilde{\mathcal{H}}_\varepsilon := \{\tilde{\varphi}(F+g\cdot G) : \tilde{\varphi}\in\tilde{L}_\varepsilon , \, \tilde{g}\in\tilde{\mathcal{G}}_\varepsilon \}.\]
	We claim that for every $H=\varphi(F+g\cdot G)\in\mathcal{H}$ there is $\tilde{H}=\tilde{\varphi}(F+\tilde{g}\cdot G)\in\tilde{\mathcal{H}}$ such that $\|H-\tilde{H}\|_{L^2(\mu_N)}\leq \varepsilon$.
	It this is true, then 
	\begin{align}
		\label{eq:covering.lipschitz}
	\begin{split}
	&\mathcal{N}(\mathcal{H},\| \cdot \|_{L^2(\mu_N)},\varepsilon)
	\leq \mathrm{card}(\tilde{\mathcal{H}}_\varepsilon)
		\leq \mathrm{card}(\tilde{L}_\varepsilon)\mathrm{card}(\tilde{\mathcal{G}}_\varepsilon)\\
		&\quad\leq  \exp\Big( \frac{C}{(\varepsilon^{(p+1)/p}/ \| J \|_{L^{2p}(\mu_N)})\wedge 1} \Big) \cdot \Big( \frac{C}{(\varepsilon^{(p+1)/p}/ \| J \|_{L^{2p}(\mu_N)}  )\wedge 1} \Big)^e
\end{split}
	\end{align}
	for every $\varepsilon>0$, where the last inequality holds by \eqref{eq:covering.card.L}, \eqref{eq:covering.card.G} and the choice of $a_\varepsilon$ in \eqref{eq:cover.a.eps}.
	
	To prove this claim, let $\tilde{\varphi}\in \tilde{L}_\varepsilon$ be such that $\sup_{t\in [-C_0a_\varepsilon,C_0a_\varepsilon]} |\varphi(t)-\tilde{\varphi}(t)|\leq \varepsilon/3$, $g\in\tilde{\mathcal{G}}_\varepsilon$ such that $|g-\tilde{g}|\leq \varepsilon/(3C_0a_\varepsilon)$, and write
	\begin{align}
	\label{eq:covering.split}
	\| H - \tilde{H} \|_{L^2(\mu_N)}
	&\leq \| 1_{J\leq a_\varepsilon} (H - \tilde{H})\|_{L^2(\mu_N)} +\| 1_{J>a_\varepsilon} (H - \tilde{H}) \|_{L^2(\mu_N)}.
	\end{align}
	To estimate the first  term in the right hand side of \eqref{eq:covering.split}, recall that $\tilde{\varphi}$ is $1$-Lipschitz and hence
	\begin{align*}
	\| 1_{J\leq a_\varepsilon} (H - \tilde{H})\|_{L^2(\mu_N)} 
	&\leq \| 1_{J\leq a_\varepsilon} (\varphi(F+g\cdot G) - \tilde{\varphi}(F+g\cdot G) ) \|_{L^2(\mu_N)} \\
	&\qquad + \|1_{J\leq a_\varepsilon} (\tilde{\varphi}(F+g\cdot G)-\tilde{\varphi}(F+\tilde{g} \cdot G) )\|_{L^2(\mu_N)} \\
	&\leq \sup_{t\in[-C_0a_\varepsilon, C_0 a_\varepsilon ]} |\varphi(t)-\tilde{\varphi}(t)| + C_0a_\varepsilon |g-\tilde{g}|
	\leq \frac{2\varepsilon}{3}
	\end{align*}
	by choice of $\tilde{\varphi}$ and $\tilde{g}$.
	As for the second term in the right hand side of \eqref{eq:covering.split}, first note that it is zero in case $p=\infty$, since $1_{J>a_\varepsilon} =0$ $\mu_N$-almost surely by the choice of $a_\varepsilon$ in \eqref{eq:cover.a.eps}.
	Otherwise, if $p<\infty$, recalling that $|H|,|\tilde{H}|\leq C_0 J$,  Markov's inequality and the choice of $a_\varepsilon$ imply that
	\begin{align*}
	\| 1_{J>a_\varepsilon} (H - \tilde{H}) \|_{L^2(\mu_N)}
	&\leq 2 C_0 	\| 1_{J>a_\varepsilon} J \|_{L^2(\mu_N)} \\
	&\leq  \frac{2C_0 \| J^p \|_{L^{2}(\mu_N)}}{a_\varepsilon^{p}}
	\leq \frac{\varepsilon}{3}.
	\end{align*}
	This proves our claim that $\| H - \tilde{H} \|_{L^2(\mu_N)}\leq\varepsilon$.
	\item
	It remains to argue that the set $\tilde{L}_\varepsilon$ in step (b) exists.
	To that end, denote by $L$ the set of all $1$-Lipschitz functions $\varphi\colon\mathbb{R}\to\mathbb{R}$ satisfying $\varphi(0)=0$.
	Further, for $\varphi\in L$, denote by 
	\[ \mathcal{R}(\varphi) \colon [-1,1] \to\mathbb{R}, 
		\quad t\mapsto \frac{\varphi(C_0 a_\varepsilon t) }{C_0 a_\varepsilon}\] 
	its rescaled restriction.
	Then $\mathcal{R}(L)$ consists of $1$-Lipschitz functions which are bounded by 1 and \cite[Theorem 2.7.1]{van1996weak} implies that there exists a set $R'_\varepsilon$ with cardinality at most $\exp( \frac{C}{(\varepsilon/a_\varepsilon)\wedge 1})$ such that, for every $\varphi\in\mathcal{R}(L)$ there is $\varphi'\in R'_\varepsilon$ with $\sup_{t\in[-1,1]} |\varphi(t)-\varphi'(t)|\leq \varepsilon/(6 C_0 a_\varepsilon)$.
	By the triangle inequality, there is a set $\tilde{R}_\varepsilon\subset\mathcal{R}(L)$ of the same cardinality as $R'_\varepsilon$ such that for every  $\varphi\in \mathcal{R}(L)$, there is $\tilde{\varphi}\in \tilde{R}_\varepsilon$ satisfying  $\sup_{t\in[-1,1]} |\varphi(t)-\tilde{\varphi}(t)|\leq \varepsilon/(3 C_0 a_\varepsilon)$.
	Now extend every $\varphi\in \tilde{R}_\varepsilon$ to a function with domain $\mathbb{R}$ via 
	\[ \mathcal{E}(\varphi)\colon\mathbb{R}\to\mathbb{R}, 
		\quad t\mapsto C_0 a_\varepsilon \varphi\Big( (-1) \vee \Big( \frac{t}{C_0 a_\varepsilon}\wedge 1 \Big) \Big)\]
	and note that 
	\[ \sup_{t\in[-C_0a_\varepsilon,C_0 a_\varepsilon] } |\varphi(t) - \mathcal{E}(\tilde{\varphi}(t)) |
	=C_0 a_\varepsilon \sup_{t\in[-1,1]} | \mathcal{R}(\varphi)(t) - \tilde{\varphi}(t) |.  \]
	Hence $\tilde{L}_\varepsilon:=\mathcal{E}(\tilde{R}_\varepsilon)$ is the desired set.
	\item
	The set $\mathcal{H}$ satisfies Assumption \ref{ass:pw.mb}.
	Indeed, first observe that the set of continuous functions from $\mathbb{R}$ to $\mathbb{R}$ endowed with the topology of uniform convergence on compacts\footnote{That is, w.r.t.\ to the topology induced by the metric $d(\varphi,\overline{\varphi}):=\sum_{k\geq 1} (1\wedge \sup_{t\in[-k,k]} |\varphi(t)-\overline{\varphi}(t)|) \cdot 2^{-k}$.} 
	is separable; hence the subset of $1$-Lipschitz functions is separable as well w.r.t.\ this topology.	
	The rest of the argument follows from similar arguments as presented in Lemma \ref{lem:pw.meausrable}.
	
	\item
	We use the empirical process version of Dudley's entropy integral theorem, i.e.\  Theorem \ref{thm:dudley}.
	Note that $H^\ast:=0\in\mathcal{H}$, and therefore Theorem \ref{thm:dudley} implies   
	\begin{align*}
	 &E\Big[ \sup_{u\in[0,v]} \delta_u^N\Big]
	\leq \frac{C}{\sqrt{N}} E\Big[\int_0^\infty \sqrt{ \log  \mathcal{N}(\mathcal{H},\| \cdot \|_{L^2(\mu_N)},\varepsilon)  } \,d\varepsilon \Big]\\
	&= \frac{C}{\sqrt{N}} E\Big[ \| J \|_{L^{2p}(\mu_N)} \int_0^\infty \sqrt{ \log \Big( \exp\Big( \frac{C}{\tilde{\varepsilon}^{(p+1)/p}\wedge 1}\Big) \Big( \frac{C}{\tilde{\varepsilon}^{(p+1)/p}\wedge 1}\Big)^e \Big) } \,d\tilde{\varepsilon} \Big],
	\end{align*}
	where the last line followed from using \eqref{eq:covering.lipschitz} and substituting $\varepsilon$ by $\tilde{\varepsilon}= \varepsilon / \| CJ^p \|_{L^2(\mu_N)}^{1/p}$.
	It remains to notice that the (now deterministic) integral over $d\tilde{\varepsilon}$ is finite.
	Moreover, Jensen's inequality implies $E[ \| J \|_{L^{2p}(\mu_N)}]\leq \|J\|_{L^{2p}(\mu)}$ and the latter term is finite by assumption.

	In conclusion, we have shown \eqref{eq:estimate.delta.u.general.risk} and the proof is complete.
	\qedhere
	\end{enumerate}
\end{proof}

\begin{proof}[Proof of Theorem \ref{thm:main.hedging.on.lp}]
	Recall the definition of $M:=\| J \|_{L^p(\mu)}$ and $M_N:=\| J \|_{L^p(\mu_N)}$ given in \eqref{eq:def.M.MN}.
	As in the proof of Theorem \ref{thm:main.oce.hedging}  we set
	\[\Delta_N:=\sup_{g\in\mathcal{G}} \Big| \rho^\mu(F+g\cdot G) -\rho^{\mu_N}(F+g\cdot G) \Big| \] and consider both terms in
	\[ E^\ast[\Delta_N]
	=E^\ast[\Delta_N 1_{M_N\leq M+1} ] + E^\ast[\Delta_N 1_{M_N> M+1} ]  \]
	separately (note that linearity of the outer expectation holds here because $\{M_N\leq M+1\}$ is a measurable set).
	\begin{enumerate}[(a)]
	\item
	As $\mathcal{G}$ is bounded, we have $\|F+g\cdot G\|_{L^p(\mu)}\leq C M$.
	Therefore, by Lemma \ref{lem:growth.Gamma}, there exists some $b$ such that
	\[ \rho^\mu(F+g\cdot G)=\sup_{\gamma \in\mathcal{M} \text{ s.t.\ } \beta(\gamma)\leq b}  \Big( \int_{[0,1)} \mathrm{AVaR}_u^\mu(F+ g\cdot G )\,\gamma(du) -\beta(\gamma)\Big) \]
	for all $g\in\mathcal{G}$.
	Possibly making $b$ larger, the same reasoning implies that, on the set $M_N\leq M+1$, the same representation holds true if $\mu$ is replaced by $\mu_N$.
	Recalling the definition of $\delta^N$ in \eqref{eq:def.delta.u} and the definition of $\Gamma_b$ given in Lemma \ref{lem:growth.Gamma}, we can write
	\begin{align*}
	\Delta_N 1_{M_N\leq M+1}
	&\leq \sup_{\gamma \in\mathcal{M} \text{ s.t.\ } \beta(\gamma)\leq b} \int_{[0,1)} \delta_u^N\,\gamma(du) \\
	&\leq \sum_{n\geq 1} \Gamma_b(I_n) \sup_{u\in I_n}\delta_u^N,
	\end{align*}
	where $I_n:=[1-2^{-n+1},1-2^{-n})$ for every $n$, that is, $I_1=[0,1/2)$, $I_2=[1/2,3/4)$ and so forth.

	Now estimate $\Gamma_b(I_n)\leq C 2^{-n/q}$ by means of Lemma \ref{lem:growth.Gamma} and $E[\sup_{u\in I_n}\delta_u^N]\leq C (2^n\sqrt{N}^{-1}) \wedge 2^{n/2p}$ 
	by means of Lemma \ref{lem:control.delta.general.risk.lp}.
	Then, an application of Lemma \ref{lem:strange.sums.general} implies that
	\begin{align*}
	E^\ast[	\Delta_N 1_{M_N\leq M+1} ]
	&\leq C \sum_{n\geq 1}  2^{-n/q} \Big( \frac{2^n}{\sqrt{N}} \wedge 2^{n/2p} \Big)\\
	&\leq \frac{C}{\sqrt{N}^\frac{1/q-1/2p}{1-1/2p}} \vee \frac{C}{\sqrt{N}}
	\leq \frac{C}{\sqrt{N}^\frac{1/q-1/2p}{1-1/2p}}
	\end{align*}
	where the last inequality holds as $\frac{1/q-1/2p}{1-1/2p}\in(0,1)$.
	\item
	The second term is controlled in a similar way as in the proof of Theorem \ref{thm:main.oce.hedging}, namely we first estimate
	\begin{align*}
	E^\ast[\Delta_N 1_{M_N>M+1}]
	&\leq E^\ast[\Delta_N^2]^{1/2} P[M_N>M+1]^{1/2} \\
	&\leq \frac{C E^\ast[\Delta_N^2]^{1/2}}{\sqrt{N}}.
	\end{align*}
	It therefore remains to check that $E^\ast[\Delta_N^2]\leq C$.
	In fact, if $p=\infty$ then $M_N\leq M$ almost surely and there is nothing left to prove.
	So assume that $p<\infty$.
	Using monotonicity of $\rho$ and the fact that $\mathcal{G}$ is bounded, this boils down to checking that $E[\rho^{\mu_N}(CJ)^2]\leq  C$.
	To that end, by definition of $w_p$ and as $J\geq 1$, one has that
	\[\rho^{\mu_N}(CJ)
	\leq w_p(C\|J\|_{L^p(\mu_N)})
	\leq w_p \Big( C \frac{1}{N}\sum_{n\leq N} J(S_n)^p \Big).\]
	By convexity of $x\mapsto w_p(x)^2$ we may further estimate
	\begin{align*}
	E^\ast[\rho^{\mu_N}(CJ)^2]
	&\leq \frac{1}{N}\sum_{n\leq N} E\Big[ w_p\Big( C J(S_n)^p \Big)^2 \Big] \\
	&=\int_{\mathcal{X}} w_p(C J(x)^p)^2\,\mu(dx)
	\end{align*}
	and the last term is finite by assumption.
	\end{enumerate}
	Combining both steps completes the proof.
\end{proof}

\section{Deviation inequalities}
\label{sec:deviation}

In the following, we prove (a generalization of) part \eqref{item:oce.intro.deviation} of Theorem \ref{thm:main.hedging.on.Linfty.intro} and part \eqref{item:oce.intro.deviation} of Theorem \ref{thm:main.oce.hedging.intro} stated in Section \ref{sec:main.results.}.

\begin{theorem}
\label{thm:general.deviation}
	Assume that $F$ and $G$ are bounded functions and that the set $\mathcal{G}$ is bounded.
		Moreover, let $q\in(1,\infty)$ and assume that $\rho$ is $q$-regular.
		Then there are constants $c,C>0$ such that
	\[ P^*\Big[ \sup_{g\in\mathcal{G}}|\rho^\mu(F+g\cdot\mathcal{G})-\rho^{\mu_N}(F+g\cdot\mathcal{G})|\geq\varepsilon \Big]
	\leq C\exp\Big(-cN \varepsilon^{2q} \Big)\]
	for all $\varepsilon>0$ and $N\geq 1$.
\end{theorem}
\begin{proof}
\begin{enumerate}[(a)]
	\item
	In a first step, recall that $F$, $G$, and $\mathcal{G}$ are bounded, hence there is a constant $a$ such that $|F+g\cdot G|\leq a$ for all $g\in\mathcal{G}$.
	As the optimal $m$ in the definition of the average value at risk is given by a respective quantile,
	it follows that 
	\begin{align}
	\label{eq:avar.bounded.mu}
	\mathrm{AVaR}^\mu_u(F+g\cdot G)
	=\inf_{|m|\leq a} \frac{1}{1-u} \int_{\mathcal{X}}(F+g\cdot G -m)_+ +(1-u)m \,\mu(dx)
	\end{align}
	for every $u\in[0,1)$ and $g\in\mathcal{G}$, and \eqref{eq:avar.bounded.mu} remains true if $\mu$ is replaced by $\mu_N$.
	Further, as $\int_{\mathcal{X}} (1-u)m\,(\mu-\mu_N)(dx)=0$ for all $m\in\mathbb{R}$ and $u\in[0,1)$, this implies that
	\begin{align}
	\label{eq:deviation.avar.estimate}
	\big| \mathrm{AVaR}^\mu_u(F+g\cdot G) - \mathrm{AVaR}^{\mu_N}_u(F+g\cdot G) \big|
	\leq \frac{ \delta^N_0}{1-u},
	\end{align}
	where we set
	\begin{align*}
	\delta^N_0
	&:= \Big| \sup_{H\in\mathcal{H}} \int_{\mathcal{X}} H(x) \,(\mu-\mu_N)(dx)\Big|\quad\text{and}\\
	\mathcal{H}
	&:=\{ (F+g\cdot G -m )_+ : |m|\leq  a \text{ and }  g\in\mathcal{G} \}.
	\end{align*}
	Note that, by Lemma \ref{lem:pw.meausrable}, the set $\mathcal{H}$ satisfies Assumption \ref{ass:pw.mb}.
	\item
	In a second step, notice that the same arguments (again, actually simpler as $J$ is bounded) as in the proof of Theorem \ref{thm:main.hedging.on.Linfty.intro} imply that there is some $b>0$ such that the supremum over $\gamma\in\mathcal{M}$ in the spectral representation \eqref{eq:rho.spectral.rep} of $\rho$ can be restricted to those $\gamma$ for which $\beta(\gamma)\leq b$.
	This implies
	\begin{align*}
	&\big| \rho^\mu(F+g\cdot F) - \rho^{\mu_N}(F+ g\cdot G)\big| \\
	&\leq \sup_{\gamma\in\mathcal{M} \text{ s.t.\ } \beta(\gamma)\leq b}  \int_{[0,1)} |\mathrm{AVaR}^\mu_u(F+g\cdot G)-\mathrm{AVaR}^{\mu_N}_u(F+g\cdot G)|\,\gamma(du)\\
	&\leq \sum_{n\geq 1} \Gamma_b(I_n) 
	 \Big( \sup_{u\in I_n} \frac{\delta^N_0}{1-u} \wedge C_1 \Big) 
	\end{align*}
	where $I_n:=[1-2^{-n+1},1-2^{-n})$ for every $n$ and the constant $C_1$ appears since $F$, $G$, and $\mathcal{G}$ are bounded.
Without loss of generality we may assume that $C_1\geq 1$.
		Then, estimating $\Gamma_b(I_n)\leq C_2 2^{-n/q} $ by Lemma \ref{lem:growth.Gamma}, we obtain
	\begin{align*}
	\begin{split}
	\sup_{g\in\mathcal{G}}\big| \rho^\mu(F+g\cdot F) - \rho^{\mu_N}(F+ g\cdot G)\big|
	&\leq C_1C_2 \sum_{n\geq 1} 2^{-n/q}
		 \Big( (2^n \delta^N_0) \wedge 1\Big)  \\
	&\leq C C_1C_2 \Big( (\delta^N_0)^{1/q}\vee \delta^N_0 \Big),
	\end{split}
	\end{align*}
	for all $N\geq 1$ almost surely, where the last inequality follows from Lemma \ref{lem:strange.sums.general}.
	Finally as $\delta^N_0\leq C_3$ almost surely, we conclude that
	\begin{align}
	\label{eq:general.deviation.prelimiary}
	\sup_{g\in\mathcal{G}}\big| \rho^\mu(F+g\cdot F) - \rho^{\mu_N}(F+ g\cdot G)\big|
	\leq C (\delta^N_0)^{1/q}.
	\end{align}
	\item
	In a last step, it remains to estimate $\delta^N_0$.
	By Corollary \ref{cor:covering.compact} one has that
	\[ \mathcal{N}(\mathcal{H}, \|\cdot\|_\infty,\varepsilon)
	\leq \Big(\frac{C}{\varepsilon}\Big)^{e+1}\vee 1\]
	for all $\varepsilon>0$.
	Hence, since $\mathcal{H}$ satisfies Assumption \ref{ass:pw.mb}, Theorem \ref{thm:concentration.inequality} implies that 
	\[ P[ \delta_0^N \geq \varepsilon ]
	\leq  C \exp\Big(-\frac{N\varepsilon^2}{C}\Big) \]
	for all $\varepsilon>0$ and $N\geq 1$.
	The proof is completed by plugging the last estimate into equation \eqref{eq:general.deviation.prelimiary}.
	\qedhere
	\end{enumerate}
\end{proof}

\begin{theorem}
\label{thm:oce.deviation}
	Assume that $F$ and $G$ are bounded functions, that the set $\mathcal{G}$ is bounded, and let $\rho=\mathrm{OCE}$ be the optimized certainty equivalent risk measure.
	Then there are constants $c,C>0$ such that
	\[ P\Big[ \sup_{g\in\mathcal{G}}|\rho^\mu(F+g\cdot\mathcal{G})-\rho^{\mu_N}(F+g\cdot\mathcal{G})|\geq\varepsilon \Big]
	\leq C\exp\Big(-cN\varepsilon^2\Big)\]
	for all $\varepsilon>0$ and $N\geq 1$.
\end{theorem}
\begin{proof}
	The proof is similar to the one given for Theorem \ref{thm:general.deviation} and we shall keep it short.
	By Lemma \ref{lem:oce.compact.estimate} one has
	\[|\rho^\mu(F+g\cdot G)-\rho^{\mu_N}(F+g\cdot G)|
	\leq \sup_{H\in\mathcal{H}} \Big| \int_{\mathcal{X}} H(x) \,(\mu-\mu_N)(dx)\Big|
	=:\delta^N_0 \]
	almost surely, for the set
	\[ \mathcal{H}:=\{ l(F+g\cdot G-m)+m : g\in\mathcal{G} \text{ and } |m|\leq a\}\]
	with $a$ such that $|F + g\cdot G| \le a$ for all $g\in\mathcal{G}$.
	By Lemma \ref{lem:pw.meausrable}, the set $\mathcal{H}$ satisfies Assumption \ref{ass:pw.mb}.
	Thus, an application of Theorem \ref{thm:concentration.inequality} again implies that $P[\delta_0^N\geq \varepsilon]\leq C\exp(-cN\varepsilon^2)$ for some constants $c,C>0$.
	This concludes the proof.
\end{proof}

\section{Sharpness of rates}
\label{sec:sharpness}

Whenever investigating average errors involving a (linear) dependence on i.i.d.\ phenomena, the central limit theorem assures that the $1/\sqrt{N}$ rate cannot be improved.
Indeed, take for instance $\rho(X):=E[X]=\mathrm{AVaR}_0(X)$.
Then, if $\mu$ is a probability on $[0,1]$ and $F$ is a (bounded) function which is equal to the identity on $[0,1]$, one  has that
\[ \rho^{\mu_N}(F)
=\frac{1}{N}\sum_{n\leq N} F(S_n)
\text{ approximately has the distribution } \mathcal{N}\Big(\rho^\mu(F), \frac{\mathrm{Var}(F(S)) }{N}\Big) \]
for large $N$ by the central limit theorem, where $\mathcal{N}$ denotes the normal distribution and $\mathrm{Var}(F(S))$ is the variance of $F(S)$.
In particular $E[| \rho^{\mu}(F)- \rho^{\mu_N}(F)|]$ asymptotically behaves like $\sqrt{\mathrm{Var}(F(S)) /N}$ and $P[| \rho^{\mu}(F)- \rho^{\mu_N}(F)|\geq \varepsilon]$ asymptotically behaves like $2 \Phi(-\varepsilon^2 N / \mathrm{Var}(F(S)))$ where $\Phi$ is the cumulative distribution function of the standard normal distribution.
We refer to \cite{Beu-Zaeh10,belomestny2012central,Chen08} for central limit theorems for risk measures. 

\vspace{0.5em}

In comparison to the above $1/\sqrt{N}$ rate, the rates obtained for general risk measures e.g.\ in Theorem \ref{thm:main.hedging.on.Linfty.intro} are worse.
As the proofs are presented, they depend on the notion of regularity of the risk measure given in Definition \ref{def:pareto} and this section is devoted to showing the necessity of regularity; we shall prove Proposition \ref{prop:non.asymptotic.no.rates.into}.
To that end, to ease the notation, for probabilities $\mu$ on $\mathbb{R}$ with bounded support, we shall write
\[\rho(\mu):=\rho(X) \quad\text{where }X\sim\mu.\]

\begin{remark}
\label{rem:rates.without.lebesque}
	Without the assumption that $\rho$ is $q$-regular, the proof of Proposition \ref{prop:non.asymptotic.no.rates.into} becomes rather trivial:
	take $\rho(X):=\mathop{\mathrm{ess.sup}} X$ and let $\mu$ be some probability with support $[0,1]$.
	As $\rho(\mu_N)=\max_{n\leq N} X_n$ (where $(X_n)$ is an i.i.d.\ sample of $\mu$) one has
	\[P [|\rho(\mu)-\rho(\mu_N)|\geq\varepsilon ]
	= P\Big[\max_{n\leq N} X_n \leq 1-\varepsilon\Big]
	=\mu([0,1-\varepsilon])^N.\]
	For suitable choices of $\mu$, the latter term can converge arbitrary slow to zero.
	Therefore $E[|\rho(\mu)-\rho(\mu_N)|]=\int_0^1 \mu([0,1-\varepsilon])^N\,d\varepsilon$ converges arbitrary slow as well.
\end{remark}

The proof of Proposition \ref{prop:non.asymptotic.no.rates.into} below mimics the idea of Remark \ref{rem:rates.without.lebesque} while simultaneously enforcing regularity of $\rho$.
To ease notation, denote by
\begin{align}
\label{eq:def.mu.p.m}
\mathrm{Ber}(p):=(1-p)\delta_0 + p\delta_1
\end{align}
the Bernoulli distribution with parameter of success $p\in[0,1]$.
Then, for $\mu=\mathrm{Ber}(p)$, the empirical measure $\mu_N$ of $\mu$ satisfies
\begin{align}
\label{eq:relation.mu.p.m.empirical}
\mu_N\equiv\mathrm{Ber}(p)_N = \mathrm{Ber}(\widehat{p}_N) \quad\text{where}\quad \widehat{p}_N:=\frac{1}{N}\sum_{n\leq N} X_n
\end{align}
(almost surely) where $(X_n)$ are i.i.d.\ $\mathrm{Ber}(p)$ distributed.
This simple formula is actually the reason why we stick to the Bernoulli distribution, as computations become a lot easier.

We start with two simple lemmas and leave their simple proofs to the reader.

\begin{lemma}
\label{lem:avar.bernulli}
	Let $p\in(0,1)$.
	Then
	\[ \mathrm{AVaR}_u(\mathrm{Ber}(p))=  \frac{p}{1-u} \wedge 1 \]
	for every $u\in[0,1)$.
\end{lemma}

\begin{lemma}
\label{lem:lemma.maximizer.for.sharpness}
	It holds that
	\[ \sup_{x\geq 1}\Big( (1-x^{-\delta}) a + x^{-\delta} \big( (a x)\wedge 1 \big) \Big)
	=(1-a^{\delta}) a + a^\delta \]
	for every $a\in[0,1]$ and $\delta>0$.
\end{lemma}

\begin{proof}[Proof of Proposition \ref{prop:non.asymptotic.no.rates.into}]
	For shorthand notation, set $\delta:=1/q$.
	Define  $\rho\colon L^\infty\to\mathbb{R}$ by
	\begin{align}
	\label{eq:rho.rates}
	\rho(X):=\sup_{x\geq 1} \Big( (1-x^\delta) \mathrm{AVaR}_0(X) + x^{-\delta} \mathrm{AVaR}_{1-1/x}(X) \Big).
	\end{align}
	As $\mathrm{AVaR}$ is a law--invariant coherent risk measure, $\rho$ inherits all those properties.
	
	To check that $\rho$ is $(q+\varepsilon)$-regular for every $\varepsilon>0$, fix such $\varepsilon$ and  denote by $X^\ast$ a random variable with the Pareto distribution with scale parameter 1 and shape parameter $q+\varepsilon$.
	Then $X^\ast$ has finite $q$-th moment and the definition of $\rho$ together with Lemma \ref{lem:avar.moments} imply that 
	\[ \rho(X^\ast\wedge n)
	\leq \| X^\ast \|_{L^{q}} \sup_{x\geq 1} \Big( (1-x^{-\delta}) 1 + x^{-\delta} x^{1/q} \Big)
	<\infty\]
	for all $n\in\mathbb{N}$.
	As the right hand side does not depend on $n$, this shows that $\rho$ is $(q+\varepsilon)$-regular.

	Now let $p_N:=1/N$ and let $(X_n^N)$ be an i.i.d.\ sample of $\mathrm{Ber}(p_N)$, that is, $P[X_n^N=1]=p_N=1/N$ for all $n$ and $N$.
	Further recall that the empirical measure of $\mathrm{Ber}(p_N)$ is $\mathrm{Ber}(\widehat{p}_N)$ where $\widehat{p}_N:= \frac{1}{N}\sum_{n\leq N} X_n^N$.
	We will show that
	\[ \rho(\mathrm{Ber}(p_N)) -E[\rho(\mathrm{Ber}(\widehat{p}_N))]
	\geq  \frac{ p_N^{ \delta} }{C} \]
	for all $N$.
	Using the triangle inequality, this clearly implies the statement of the proposition.

	By Lemma \ref{lem:avar.bernulli} and Lemma \ref{lem:lemma.maximizer.for.sharpness} we compute
	\begin{align*}
	\rho(\mathrm{Ber}(p_N))
	&=\sup_{x\geq 1} \Big( (1-x^{-\delta}) p_N + x^{-\delta} \big( (xp_N)\wedge 1 \big)  \Big) \\
	&= (1-p_N^\delta )p_N + p_N^\delta
	\end{align*}
	and similarly
	\begin{align*}
	\rho(\mathrm{Ber}(\widehat{p}_N))
	&=(1-\widehat{p}_N^\delta)\widehat{p}_N + \widehat{p}_N^\delta.
	\end{align*}
	Now recall that $E[\widehat{p}_N]=p_N$ and, by Jensen's inequality, $E[\widehat{p}_N^\delta\widehat{p}_N]\geq p_N^\delta p_N$; hence
	\begin{align*}
	\rho(\mathrm{Ber}(p_N)) -E[\rho(\mathrm{Ber}(\widehat{p}_N))]
	\geq p_N^\delta-E[\widehat{p}_N^\delta].
	\end{align*}

	For the set
	\[A_N:=\{\widehat{p}_N=0 \}
	=\{ X_n^N = 0\text{ for all } n\leq N\},\]
	one computes
	\begin{align*}
	P[A_N]
	&= (1-p_N)^N
	=\exp\Big(N\log\Big(1-\frac{1}{N}\Big)\Big)
	\geq \exp(-2)
	\end{align*}
	for $N\geq 2$.
	Moreover $E[\widehat{p}_N^\delta]= E[\widehat{p}_N^\delta 1_{A_N^c}]$ and an application of H\"older's inequality (with exponents $p=1/\delta$ and $q=1/(p-1)=1/(1-\delta)$) gives
	\begin{align*}
 	E[\widehat{p}_N^\delta]
	&\leq E[\widehat{p}_N]^\delta P[A_N^c]^{1-\delta} \\
	& \leq p_N^\delta  \big(1-\exp(-2) \big)^{1-\delta}
	=:p_N^\delta c
	\end{align*}
	for all $N\geq 2$.
	Here we also used that $E[\widehat{p}_N]=p_N$ and the previous computation for (the limit of) $P[A_N]$.

	In particular
	\begin{align*}
	\rho(\mathrm{Ber}(p_N)) -E[\rho(\mathrm{Ber}(\widehat{p}_N))]
	&\geq p_N^\delta (1 - c)
	\end{align*}
	for all $N\geq 2$.
	As $c\in(0,1)$, this completes the proof (considering the case $N=1$ separately).
\end{proof}

\begin{remark}
\label{rem:fatou.lebesque}
	In the theory of risk measures two continuity properties are often considered:
	the Fatou property and the stronger Lebesgue property.
	We refer the unfamiliar reader to \cite[Section 4.2]{foellmer01}. 
	A result of Jouini, Schachermayer and Touzi \cite{Jou-Sch-Tou} assures that every law--invariant risk measure automatically satisfies the Fatou property, and it is easy to see that a $q$-regular law--invariant risk measure satisfies the Lebesgue property.
	
	Small modifications in the proof of Proposition \ref{prop:non.asymptotic.no.rates.into} actually give the existence of a law--invariant risk measure which satisfies the Lebesgue property but for which no polynomial convergence rate hold true.
\end{remark}

\section{Additional proofs}
\label{sec:aux}

\subsection{Remaining proofs for Theorem \ref{thm:main.hedging.on.Linfty.intro}}

	We finally provide the proof of Theorem \ref{thm:main.hedging.on.Linfty.intro} for the case that $\rho$ is the shortfall risk measure.

	\begin{enumerate}[(a)]
	\item
	Define the function $J\colon\mathbb{R}\to\mathbb{R}$ by
	\[J(m):=\inf_{g\in\mathcal{G}} \int l( F + g\cdot G -m )\,\mu(dx)\]
	and in the same way define the (random) function $J_N$ with $\mu$ replaced by $\mu_N$.
	Further let $a\geq 0$ such that $|F+g\cdot G|\leq a$ for every $g\in\mathbb{R}$.
	Then $|\pi^{\mu}(F)|\leq a$, or, in other words
	\[ \pi^\mu(F)=\inf\{ m\in [a,a] : J(m)\leq 0  \}.\]
	The same is true if $\mu$ is replaced by $\mu_N$ and $J$ by $J_N$ (almost surely).
	\item
	We claim that there is $c>0$ such that  $J(\tilde{m})\leq  J(m) -c(\tilde{m}-m)$ for all $m,\tilde{m}\in[-a,a]$ with $m\leq \tilde{m}$.
	Indeed, as $l$ is convex and strictly increasing, its (right) derivative $l'$ is strictly positive.
	Now let $g\in\mathcal{G}$ be optimal for $J(m)$ (for notational simplicity, otherwise use some $\varepsilon$-optimal $g$), that is, $J(m)=\int l(F+g\cdot G-m)\,d\mu$.
	The fundamental theorem of calculus then implies
	\begin{align*}
	&J(\tilde{m})
	\leq \int l(F+g\cdot G - \tilde{m})\,d\mu\\
	&=\int l(F+g\cdot G - m) - (\tilde{m}-m)\int_0^1 l'(F+g\cdot G -m + t(\tilde{m}-m))\,dt \,d\mu.
	\end{align*}
	The term inside the the second integral is larger than $c:=\inf_{ |t|\leq 2a } l'(t)>0$.
	So $J(\tilde{m})\leq  J(m) -c(\tilde{m}-m)$, which is what we claimed.
	\item
	We claim that $J$ and $J_N$ are continuous.
	Indeed, this is an easy consequence of the continuity of $(m,g)\mapsto \int l(F+g\cdot G-m)\,d\mu$ together with the fact that $\mathcal{G}$ it relativity compact (similarly for $J_N$); we spare the details.
	\item
	Step (b) in particular implies that $J$ is strictly increasing.
	Combining this with the continuity of $J$ yields that $\pi^\mu(F)$ is the unique number satisfying $J(\pi^\mu(F))= 0 $.
	Similarly, $\pi^{\mu_N}(F)$ is the unique number satisfying $J_N(\pi^{\mu_N}(F))= 0 $ and therefore
	\begin{align*}
	| J(\pi^{\mu_N}(F)) - J(\pi^{\mu}(F)) |
	&= |J(\pi^{\mu_N}(F)) - J_N(\pi^{\mu_N}(F))|\\
	&\leq \sup_{|m|\leq a} |J(m)-J_N(m)|
	=:\Delta_N.
	\end{align*}
	Making use of step (a), this implies $|\pi^{\mu_N}(F)-\pi^\mu(F)|\leq c \Delta_N$ and so it remains to gain control over $\Delta_N$.
	As
	\begin{align*}
	\Delta_N
	&\leq \sup_{H\in\mathcal{H}} \Big| \int H\,d(\mu-\mu_N) \Big| \quad\text{for}\\
	\mathcal{H}
	&:=\{ l(F+ g\cdot G - m ) : |m|\leq a \text{ and } g\in\mathcal{G} \},
	\end{align*}
	we can use Lemma \ref{lem:pw.meausrable}, Corollary \ref{cor:covering.compact},  and Dudley's theorem as in the proof of Theorem \ref{thm:main.oce.hedging.intro} to obtain $E[\Delta_N]\leq C/\sqrt{N}$ for all $N\geq 1$.
	Similarly, Corollary \ref{cor:covering.compact}  and the arguments given for the proof of Theorem \ref{thm:general.deviation} imply that $P[\Delta_N\geq \varepsilon]\leq C\exp(-cN\varepsilon^2)$ for all $\varepsilon>0$, $N\geq 1$, where $c>0$ is some (new) small constant.
	This completes the proof.
	\end{enumerate}
	
\subsection{The proof of Proposition \ref{prop:utility}}

We only sketch the proof of Proposition \ref{prop:utility}, as it is very similar to that of Theorem \ref{thm:main.hedging.on.Linfty.intro} on the optimized certainty equivalents.
The only difference is the absence of the component $m$ (in the definition of $\mathrm{OCE}$), which actually makes the proof even simpler.
In particular, we have
\[ \mathcal{N}\Big(\big\{U( F+ g\cdot G) : g\in\mathcal{G}\big\},\|\cdot\|_{\infty},\varepsilon\Big)\leq \Big(\frac{C}{\varepsilon}\Big)^{e}\vee 1\]
for all $\varepsilon>0$ by Corollary \ref{cor:covering.compact}.
To conclude the proof, copy the arguments given for the proofs of Theorem \ref{thm:main.oce.hedging} and Theorem \ref{thm:oce.deviation}.

\vspace{2em}
\noindent
\textsc{Acknowledgments:}
%

The authors would like to thank Patrick Cheridito as well as the Associate Editor and two Referees for extraordinarily helpful comments.
Daniel Bartl is grateful for financial support through the Vienna Science and Technology Fund (WWTF) project MA16-021 and the Austrian Science Fund (FWF) trough projects ESP-31 and P34743,
Ludovic Tangpi is supported by the NSF grant DMS-2005832 and NSF CAREER award DMS-2143861.


\providecommand{\bysame}{\leavevmode\hbox to3em{\hrulefill}\thinspace}
\providecommand{\MR}{\relax\ifhmode\unskip\space\fi MR }
\providecommand{\MRhref}[2]{%
  \href{http://www.ams.org/mathscinet-getitem?mr=#1}{#2}
}
\providecommand{\href}[2]{#2}

\appendix
\section{Supplementary results}

In this appendix we provide an additional result pertaining to the boundedness assumption on $\mathcal{G}$.
Recall that $\mathcal{G}$ is said to be bounded as a subset of $\mathbb{R}^e$ equipped with the Euclidean norm.
Further, we state two results from empirical process theory that we used in this article, and comment on the measurability of the plug-in estimator.

\subsection{The set $\mathcal{G}$ needs to be bounded}
Our set up also includes the case of risk based hedging, in which case one would rather write
\[ \pi^\mu(F)=\inf\Big\{ m\in\mathbb{R} : \text{there is some $g\in\mathcal{G}$ such that } \rho^\mu(F-m+g\cdot G)\leq 0\Big\}. \]
(This expression follows from additivity on the constants of $\rho^\mu$).

In prose, $\pi^\mu(F)$ is the minimal capital $m$ needed such that, possibly after trading, the loss $F$ reduced by $m$ becomes acceptable.
In this setting one would typically not restrict to bounded strategies, that is, one would take $\mathcal{G}=\mathbb{R}^e$.

The goal of this section is to prove the next proposition, which states that requiring $\mathcal{G}$ to be bounded is not just a technical simplification we made, but in fact necessary.

One precaution needs to be made though: Assume for instance that $G_i=0$ for all $i$, then clearly $g\mapsto \rho^\mu(F+g\cdot G)$ does not depend on $g$ and the size of $\mathcal{G}$ does not matter.
To exclude such cases (without too much effort), we assume that $(\mu,G)$ is non-degenerate in the sense that for every $g\in\mathbb{R}^e\setminus\{0\}$ one has $\mu(g\cdot G<0)>0$.

\begin{proposition}
\label{prop:G.needs.to.be.bounded}
	Let $\rho\colon L^\infty\to\mathbb{R}$ be any law--invariant risk measure, let $F$ and each $G_i$ be bounded, and let $(\mu,G)$  be non-degenerate in the above sense.
	Assume that $\pi^\mu(F)\in\mathbb{R}$ and
	\[ E[|\pi^{\mu}(F)-\pi^{\mu_N}(F)|]\to 0 \]
	as $N\to\infty$.
	Then the set $\mathcal{G}$ needs to be bounded.
\end{proposition}
\begin{proof}
	We show the negation, namely that if $\mathcal{G}$ is unbounded, convergence cannot be true.
	To that end, let $(g^n)$ be a sequence in $\mathcal{G}$ witnessing that $\mathcal{G}$ is unbounded.
	After passing to a subsequence, there exists $g^\ast\in\mathbb{R}^e$ with $|g^\ast|=1$ such that $g^n/|g^n|\to g^\ast$.
	By assumption, $\mu(g^\ast \cdot G <0)>0$, hence there is $\varepsilon>0$ such that
	\[\mu(U)>0
	\quad\text{where}\quad
	U:=\{x\in\mathcal{X} : g^\ast\cdot G(x) < -\varepsilon\}.\]

	By definition of $\pi$ one has
	\[ \pi^{\mu_N}(F)\leq \rho^{\mu_N}(F+g^n\cdot G)  \]
	for every $n\in\mathbb{N}$.
	Moreover, it holds that
	\[F+ g^n\cdot G
	\leq \sup_{U} F + \sup_{ U} g^n\cdot G
	=: a_n
	\quad \mu_N\text{-a.s.\ on } \{\mu_N(U)=1\} \]
	for every $n\in\mathbb{N}$.
	By assumption the first term in the definition of $a_n$ is bounded.
	Further, as  $g^n/|g^n|$ converges to $g^\ast$, one has  that
	\begin{align*}
	g^n\cdot G
	&= |g^n|\Big( g^\ast\cdot G + \Big( \frac{g^n}{|g^n|}-g^\ast\Big)\cdot G\Big) \\
	&\leq |g^n|\Big(-\varepsilon + C\Big|\frac{g^n}{|g^n|}-g^\ast\Big|\Big)
	<-\frac{|g^n|\varepsilon}{2}
	\end{align*}
	on $U$ for all large $n$.
	By monotonicity of $\rho^{\mu_N}$, this implies
	\[\rho^{\mu_N}(F + g^n\cdot G)
	\leq \rho^{\mu_N}(a_n)
	=a_n\to -\infty
	\quad\text{on } \{\mu_N(U)=1\}\]
	as $n\to\infty$.
	Finally, as
	\[P[\mu_N(U)=1]=1-(1-\mu(U))^N>0 \]
	for every $N\geq 1$, we conclude that $\pi^{\mu_N}(F)=-\infty$ with positive probability.
	In particular $E[|\pi^{\mu}(F)-\pi^{\mu_N}(F)|]=\infty$ for every $N\geq 1$, which proves the claim.
\end{proof}

\begin{remark}
\label{rem:G.countable}
Let us argue that our standing assumption that $\mathcal{G}$ is a countable set (which is there to circumvents issues regarding measurability) can can be made without loss of generality.
If $\mathcal{G}$ is not necessarily countable, we take a  subset $\mathcal{G}'\subset\mathcal{G}$ which is countable and dense.
If $\rho^\nu$ is a risk measure that is finite on $L^p(\nu)$ for $p\in[1,\infty]$, then it is automatically continuous w.r.t.\ $\|\cdot\|_{L^p(\nu)}$.
In particular, if $F$ and $|G|$ are in $L^p(\nu)$, then, for every $g\in\mathcal{G}$ and $(g^n)_n\subset\mathcal{G}'$ that converges to $g$, we have that $F+g^n\cdot G\to F+ g\cdot G$ in $L^p(\nu)$. 
As a consequence, under the above assumptions, 
\[ \pi^\nu(F)
=\inf_{g\in\mathcal{G}} \rho^\nu(F+g\cdot G)
=\inf_{g\in\mathcal{G}'} \rho^\nu(F+g\cdot G). \]
This shows that, in all considerations made in this paper, the set  $\mathcal{G}$ can be replaced by the set $\mathcal{G}'$.
\end{remark}

\subsection{Two inequalities from empirical process theory}
\label{sec:empirical.processes}

Recall that $S,(S_n)_{n\geq 1}$ are i.i.d.\ random variables taking their values in a Polish  space $\mathcal{X}$ distributed according to $\mu$, and that $\mu_N:=\frac{1}{N}\sum_{n\leq N} \delta_{S_n}$ is the associated empirical measure.
Moreover, let $\mathcal{F}$ be a set of measurable, $\mu$-square integrable functions from $\mathcal{X}$ to $\mathbb{R}$, and recall that $\mathcal{N}( \mathcal{F}, \|\cdot\|_{L^p(\mu)},\varepsilon)$ is the covering number of $\mathcal{F}$ w.r.t.\ the $L^p$-norm at scale $\varepsilon$.
Finally, we present result only under the following assumption regarding measurability.

\begin{assumption}
\label{ass:pw.mb}	
There is a countable set $\mathcal{F}'\subset\mathcal{F}$ such that, 
for every $f\in\mathcal{F}$ there is a sequence $(f_n)_n$ in $\mathcal{F}'$ such that $f_n\to f$ pointwise and in $L^2(\mu)$.
\end{assumption} 

Note that Assumption \ref{ass:pw.mb} implies that 
\begin{align*}
\sup_{f\in\mathcal{F}} \Big|\frac{1}{N}\sum_{n\leq N} f(S_n) - E[f(S)]\Big|
=\sup_{f\in\mathcal{F}'} \Big|\frac{1}{N}\sum_{n\leq N} f(S_n) - E[f(S)]\Big|;
\end{align*} 
in particular the term on the right hand side is measurable.
Moreover, a set $\mathcal{F}$ that satisfies Assumption \ref{ass:pw.mb} is  pointwise measurable in the sense of \cite{van1996weak} (see Section 2.3.3 therein).

We first state Dudley's entropy integral theorem in the form needed here.
\begin{theorem}
\label{thm:dudley}
	Suppose that Assumption \ref{ass:pw.mb} is satisfied and let $f^\ast\in\mathcal{F}$ be an  arbitrary but fixed function in $\mathcal{F}$.
	Then we have that
	\begin{align*}
	&E\Big[\sup_{f\in\mathcal{F}} \Big| \frac{1}{N}\sum_{n\leq N} f(S_n) - E[f(S)] \Big| \Big] \\
	&\leq \frac{C}{\sqrt{N}} \Big(  E[ f^\ast(S)^2]^{\frac{1}{2}} +  E\Big[ \int_{0}^\infty \sqrt{\log \mathcal{N}(\mathcal{F},\|\cdot\|_{L^2(\mu_N)},\varepsilon)} \,d\varepsilon \Big] 
	\end{align*}
	for all $N\geq 1$, where $C$ is an absolute constant.
\end{theorem}
\begin{proof}
	For completeness, we provide the proof of this standard fact.
	Note that, in all of the arguments below, the set $\mathcal{F}$ can be replaced without loss of generality by $\mathcal{F}'$, ensuring that not meausrabiliy issues occur.
	By the symmetrization lemma (see \cite[Lemma 2.3.1]{van1996weak}) we have that
	\[ E\Big[\sup_{f\in\mathcal{F}} \Big| \frac{1}{N}\sum_{n\leq N} f(S_n) - E[f(S)] \Big| \Big] 
	\leq 2 E\Big[\sup_{f\in\mathcal{F}} \Big| \frac{1}{N}\sum_{n\leq N} \varepsilon_n f(S_n)\Big| \Big], \]
	where the $(\varepsilon_n)_{n}$ are i.i.d.\ random signs (i.e.\ $P[\varepsilon_n=\pm 1]=1/2$) which are stochastically independent of $(S_n)_n$.
	Now, Hoeffding's inequality \cite[Lemma 2.2.7]{van1996weak} implies that, conditionally on $(S_n)_{n}$, the process $(X_f)_{f\in\mathcal{F}}$ where $X_f:= \frac{1}{\sqrt{N}}\sum_{n=1}^N \varepsilon_n f(S_n)$ is sub-Gaussian w.r.t.\ to the $L^2(\mu_N)$ norm.
	Hence, Dudley's entropy integral theorem \cite[Corollary 2.2.8]{van1996weak} for sub-Gaussian processes applied conditionally on $(S_n)_n$ gives 
	\begin{align}
	\label{eq:dudley.conditionally}
	 E\Big[\sup_{f\in\mathcal{F}} |X_f| \Big| (S_n)_{n}\Big]
	\leq E\big[|X_{f^\ast}| \big| (S_n)_{n} \big] + C \int_0^\infty \sqrt{ \log \mathcal{N}( \mathcal{F}, \|\cdot\|_{L^2(\mu_N)},\varepsilon) } \,d\varepsilon.
	\end{align}
	Finally, as the $\varepsilon_n f^\ast(S_n)$'s are i.i.d.\ zero mean random variables, H\"older's inequality assures that
	\[ E[|X_{f^\ast}| ]
	\leq E\Big[ \Big(\frac{1}{\sqrt{N}}\sum_{n\leq N} \varepsilon_n f^\ast(S_n) \Big)^2 \Big]^{\frac{1}{2}}
	=   E[ f^\ast(S)^2]^{\frac{1}{2}} .  \]
	Thus the statement of the theorem follows by integrating \eqref{eq:dudley.conditionally}.
\end{proof}

\begin{theorem}
\label{thm:concentration.inequality}
	Suppose that Assumption \ref{ass:pw.mb} is satisfied, that there is $M$ such that $|f|\leq M$ for all $f\in\mathcal{F}$, and that
	\begin{align}
	\label{eq:deviation.covering}
	\mathcal{N}( \mathcal{F}, \|\cdot\|_{\infty},\varepsilon)
	\leq \Big( \frac{a}{\varepsilon}\Big)^b \vee 1
	\end{align}
	for all $\varepsilon>0$, where $a,b>0$ are two constants.
	Then there is a constant $C$ depending only on $a,b,M$ such that
	\[ P\Big[ \sup_{f\in\mathcal{F}} \Big| \frac{1}{N}\sum_{n\leq N} f(S_n) - E[f(S)] \Big| \geq \varepsilon \Big] 
	\leq C\exp\Big( -\frac{   N  \varepsilon^2   }{ C} \Big) \]
	for all $\varepsilon>0$ and $N\geq 1$.
\end{theorem}
\begin{proof}
	The goal is to apply \cite[Theorem 2.14.10]{van1996weak}.
	To that end, we may assume without loss of generality that $0\leq f\leq 1$ for every $f\in\mathcal{F}$.
	In a first step, note that \eqref{eq:deviation.covering} implies that
	\[ \sup_{Q} \log \mathcal{N}( \mathcal{F}, \|\cdot\|_{L^2(Q)},\varepsilon)
	\leq \frac{K}{\varepsilon}\]
	for all $\varepsilon>0$, where $K$ is a constant depending only on $a$ and $b$ and the supremum is taken over all probability measures.
	This is exactly equation (2.14.8) in \cite{van1996weak} which is the assumption needed to apply \cite[Theorem 2.14.10]{van1996weak}.
	Hence, with the notation of that theorem, we have $U=5/3<2$ so that for $\delta=1/6$ there are constants $\alpha$ and $\beta$ depending only on $K$ such that
	\[  P\Big[ \sup_{f\in\mathcal{F}} \Big| \frac{1}{N}\sum_{n\leq N} f(S_n) - E[f(S)] \Big| \geq \varepsilon \Big] 
	\leq \alpha \exp(\beta ( \sqrt{N}\varepsilon)^{U+\delta})\exp(-2  (\sqrt{N}\varepsilon)^2)
	 \]
	for every $\varepsilon>0$ and $N\geq 1$.
	Recalling that $U+\delta= 11/12 <2$, a quick computation shows that 
	\[
	\alpha \exp(\beta ( \sqrt{N}\varepsilon)^{U+\delta})\exp(-2  (\sqrt{N}\varepsilon)^2)
	\leq  C \exp\Big(-\frac{ N  \varepsilon^2}{C}\Big),
	 \]
	for some constant $C$ depending on $\alpha$ and $\beta$.
	This completes the proof.
\end{proof}

\end{document}